\documentclass[journal]{IEEEtran}

\usepackage{mathtools}
\usepackage{amssymb}
\usepackage{amsmath}
\usepackage{algorithm}
\usepackage{algpseudocode}
\usepackage{enumerate,multirow}
\usepackage{algpseudocode}
\usepackage{tikz}
\usepackage{tikz}
\usetikzlibrary{patterns,positioning,arrows,calc}
\usepackage{siunitx}
\usepackage{multirow}
\usepackage{multicol}
\usepackage{color}
\usepackage{xcolor}
\usepackage{algorithm}
\usepackage{algpseudocode}
\usepackage{graphicx}
\usepackage{subcaption}
\algrenewcommand\algorithmicrequire{\textbf{Input:}}
\algrenewcommand\algorithmicensure{\textbf{Output:}}
\usepackage{comment}

\usepackage{geometry}
\usepackage{tikz}
\usepackage{tikz}
\usetikzlibrary{patterns,positioning,arrows,calc}\usepackage{graphicx}
\usepackage{graphicx}
\usepackage{tikz}
\usetikzlibrary{trees}
\date{}

\usepackage{calc}
\newcolumntype{M}[1]{>{\centering\arraybackslash}m{#1}}
\newcolumntype{N}{@{}m{0pt}@{}}

\usepackage{mdwmath}
\usepackage{blindtext}
\usepackage{eqparbox}
\usepackage{stfloats}
\usepackage{amsthm}
\usepackage[numbers,sort&compress]{natbib}
\usepackage{algorithm}
\usepackage{algpseudocode}
\usepackage{tikz}
\usetikzlibrary{patterns, arrows.meta, positioning}
\IEEEoverridecommandlockouts
\title{\Huge
Quantum LDPC and High-Rate CSS Codes\\ from Fair-Density Parity-Check Codes}
\author{\large%
Hessam Mahdavifar
\thanks{H.\ Mahdavifar is with the Department of Electrical and Computer Engineering, Northeastern University, Boston, MA 02115  (email: h.mahdavifar@northeastern.edu.)}
}

\theoremstyle{plain}

\newtheorem{theorem}{{Theorem}}
\newtheorem{lemma}[theorem]{{Lemma}}
\newtheorem{proposition}[theorem]{{Proposition}}
\newtheorem{corollary}[theorem]{{Corollary}}

\newtheorem{definition}{{Definition}}
\newcommand{\cA}{{\cal A}} 

\newcommand{\cC}{{\cal C}}

\DeclareMathAlphabet{\mathbfsl}{OT1}{ppl}{b}{it} 

\newcommand{\ba}{\mathbfsl{a}} 
\newcommand{\bbb}{\mathbfsl{b}} 

\newcommand{\bee}{\mathbfsl{e}}

\newcommand{\be}[1]{\begin{equation}\label{#1}}
\newcommand{\ee}{\end{equation}}

\renewcommand{\le}{\leqslant} 
\renewcommand{\leq}{\leqslant}
\renewcommand{\ge}{\geqslant} 
\renewcommand{\geq}{\geqslant}

\renewcommand{\Bbb}{\mathbb}

\newcommand{\Cref}[1]{Co\-ro\-lla\-ry\,\ref{#1}}

\newcommand{\Ftwo}{{{\Bbb F}}_{\!2}}

\newcommand{\deff}{\mbox{$\stackrel{\rm def}{=}$}}

\DeclareMathOperator{\wt}{wt}
\DeclareMathOperator{\rank}{rank}
\DeclareMathOperator{\supp}{supp}

\usepackage{hvlogos}
\usepackage{qtree}
\usepackage{tikz}
\usepackage{multirow}
\usepackage{hhline}
\usepackage[
    colorlinks=true,
    linkcolor=black,
    citecolor=black,
    urlcolor=black
]{hyperref}
\usepackage{pgfplots}
\pgfplotsset{compat=1.18}
\usepackage{xcolor}

\definecolor{colorPair1}{HTML}{1F77B4}
\definecolor{colorPair2}{HTML}{D62728}
\definecolor{colorPair3}{HTML}{2CA02C}

\begin{document}

\vspace{10mm}
\maketitle

\begin{abstract}
We construct quantum low-density parity-check (qLDPC) and high-rate
Calderbank--Shor--Steane (CSS) codes from our recently introduced classical fair-density parity-check
(FDPC) codes. To this end, we
introduce a structured sparsification of FDPC parity-check matrices,
which reduces their check weights while preserving the underlying combinatorial
structure and distance guarantees. Combined with the hypergraph-product
construction, this yields finite-length qLDPC codes with analytically
controlled blocklength, dimension, certified distance, and stabilizer
weight. For quantum blocklengths $N<10^5$, the constructions introduced
here span guaranteed rates from approximately $0.35\%$ to $25.8\%$ and
certified quantum distances from $12$ to $69$, with stabilizer weights
between $8$ and $16$.

In the large-blocklength regime, allowing the FDPC order and sparsified check
weight to scale moderately with blocklength yields a family of high-rate
CSS codes, which we refer to as quantum FDPC (qFDPC) codes, with rate $R_Q$ and minimum distance $D$ satisfying
\[
R_Q
=
1-O\!\left(\frac{1}{\log\log N}\right),
\ \ 
D=\Omega(N^{1/4}),
\]
and stabilizer weight $O(\log N\log\log N)$. Finally, the analytically
available FDPC weight distribution provides explicit information about the
logical operators of the resulting hypergraph-product codes. Over the
quantum erasure channel, this structure yields rigorous first-order maximum likelihood (ML)
expressions and a higher-order weight-distribution approximation to the
ML logical block error probability. This enables estimation of finite-length
operating points and the \textit{onset} of the
error-floor regime. 

To the best of our knowledge, beyond surface-code-type constructions, this is the first finite-rate qLDPC framework to provide both finite-length certified minimum-distance information and an analytical characterization of low-weight logical multiplicities.
\end{abstract}

\section{Introduction}

Quantum error correction provides the basic mechanism for protecting quantum information against noise, with stabilizer codes and the Calderbank--Shor--Steane (CSS) construction providing the linear algebraic framework underlying many modern quantum code families \cite{PhysRevA.54.1098,PhysRevA.54.4741,gottesman1997stabilizercodesquantumerror}. In particular, quantum low-density parity-check (qLDPC) codes have emerged as a compelling class \cite{MacKay_2004,PhysRevA.88.012311,TillichZemor2014,7336474,
Roffe_2020,Panteleev_2021,Breuckmann_2021} because they satisfy the desired property of sparse stabilizer measurements together with the possibility of nonvanishing coding rate and growing minimum distance \cite{Breuckmann_2021}. 

From a code construction perspective, qLDPC codes are substantially more constrained than their classical LDPC counterparts: the sparse parity-check matrices defining the $X$- and $Z$-type stabilizers must simultaneously satisfy the exact commutation condition $H_XH_Z^{\mathsf T}=0$. This constraint left open a central theoretical question in quantum coding theory: whether qLDPC codes could simultaneously have constant rate, linear minimum distance, and bounded stabilizer weight. A sequence of breakthroughs over the past decade has led to a solution for this fundamental problem. The hypergraph-product construction of Tillich and Zémor gave qLDPC codes with positive rate and minimum distance proportional to the square root of the quantum blocklength \cite{TillichZemor2014}. This $\sqrt{N}$ barrier was subsequently surpassed by fiber-bundle codes, which achieved distance $\Omega(N^{3/5}/\operatorname{polylog}N)$ \cite{HastingsHaahODonnell2021}, and by lifted-product constructions of Panteleev and Kalachev, which attained almost-linear minimum distance \cite{PanteleevKalachev2022LP}. Most notably, Panteleev and Kalachev later constructed asymptotically good qLDPC codes with both constant rate and linear minimum distance, resolving the long-standing qLDPC conjecture \cite{PanteleevKalachev2022Good}. Quantum Tanner codes also provided a closely related expander-based route to asymptotically good quantum codes with linearly growing minimum distance \cite{leverrier2022quantum}. These results establish the remarkable potential of qLDPC codes in the asymptotic regime. At the same time, they leave considerable room for constructions whose finite-length parameters, stabilizer weights, and minimum distance can be characterized or bounded analytically.

This work approaches the finite-length qLDPC design problem from a
different perspective. Instead of following the mainstream approach, i.e., starting from a sparse quantum code and
then attempting to analyze its structure, we begin with a classical code
family with sufficiently explicit combinatorial properties. Our recently introduced fair-density parity-check
(FDPC) codes provide such a starting point
\cite{mahdavifar2024high}: they have very small column weight, controlled
row weight, and a highly structured underlying graph representation. We
first exploit this structure to derive an exact ensemble-average weight
distribution for FDPC codes of arbitrary order and use it to obtain
probabilistic minimum-distance guarantees. We then introduce structured
sparsification procedures that replace the moderately dense FDPC checks by
lower-weight checks while retaining controlled column weight and preserving
the minimum-distance guarantee of the \textit{parent} FDPC code. For the finite-length
ensembles of interest, the resulting sparsified components also preserve an
analytically accessible weight distribution.

Using these sparsified FDPC codes as components of the hypergraph-product
construction leads to two complementary quantum regimes. With fixed
component degrees, we obtain finite-length qLDPC codes with bounded
stabilizer weight and analytically controlled tradeoffs among blocklength,
dimension, and certified minimum distance. In a different scaling regime,
where the FDPC order and sparsified check weight grow moderately with
blocklength, the construction yields high-rate CSS codes, which
we refer to as quantum FDPC (qFDPC) codes, with
$R_Q\to1$, $D=\Omega(N^{1/4})$, and polylogarithmic stabilizer weight.
Beyond the code parameters, the readily available FDPC weight distribution
provides additional finite-length information through the canonical
logical structure of hypergraph-product codes. Over the quantum erasure
channel, this connection yields rigorous first-order maximum likelihood (ML) expressions and a
higher-order weight-distribution approximation that can be used to estimate
finite-length operating points and predict the \textit{onset} of the
error-floor regime.

It is worth noting that the finite-length analytical access to low-weight logical operators is
not common among qLDPC constructions. Surface-code-type constructions are an important
exception, where geometric structure and weight-enumerator
methods allow explicit characterization of minimum-weight and
low-weight logical operators
\cite{Farrelly2022LocalTensor,Forlivesi2025Performance}.
Beyond this topological setting, existing qLDPC weight-distribution
analyses have largely focused on establishing or verifying minimum
distance \cite{7891006,Kasai2026GV}. To the best of our knowledge, the
framework developed here is the first finite-rate qLDPC construction to
provide both finite-length certified minimum-distance information and an
analytical characterization of low-weight logical multiplicities.

The rest of this paper is organized as follows.
Section~\ref{sec:preliminaries} reviews CSS codes, hypergraph-product
constructions, and the classical FDPC framework.
Section~\ref{sec:fdpc-extension} develops the higher-order FDPC
weight-distribution analysis and introduces the sparsification procedures
used throughout the paper.
Section~\ref{sec:qfdpc} applies these classical components to HGP
constructions, while
Section~\ref{sec:erasure-distribution} relates the FDPC weight
distribution to the logical structure of the resulting quantum codes and
develops the ML analysis for the quantum erasure channel.
Section~\ref{sec:finite-qfdpc} presents various types of new finite-length qLDPC constructions and evaluates their
parameters and performance.
Finally, Section~\ref{sec:conclusion} concludes the paper and discusses
directions for future work.

\section{Preliminaries}
\label{sec:preliminaries}

\subsection{Quantum Stabilizer and CSS Codes}

Throughout the paper, $\ker(\cdot)$,
$\operatorname{row}(\cdot)$, and $\rank(\cdot)$ denote the kernel, row
space, and binary rank of a matrix, respectively. An $[[N,K,D]]$ quantum stabilizer code is a $2^K$-dimensional subspace
$\mathcal{C}\subseteq(\mathbb{C}^2)^{\otimes N}$ specified by a collection
of mutually commuting $N$-qubit Pauli operators, called stabilizers. We denote
the group generated by these operators by $\mathcal{S}$. Then the code space is
their simultaneous $+1$ eigenspace:
\begin{equation}
\label{eq:StabilizerCode}
\mathcal{C}
=
\left\{
|\psi\rangle\in(\mathbb{C}^2)^{\otimes N}:
s|\psi\rangle=|\psi\rangle,\quad \forall s\in\mathcal{S}
\right\}.
\end{equation}
Each independent generator of $\mathcal{S}$ imposes a parity-check
constraint on the encoded state. The minimum distance $D$ is the minimum
Pauli weight of a nontrivial logical operator, namely
\begin{equation}
\label{eq:stabilizer-distance}
D
=
\min\left\{
\wt(P):
P\in\mathcal{N}(\mathcal{S})\setminus\mathcal{S}
\right\},
\end{equation}
where $\mathcal{N}(\mathcal{S})$ denotes the normalizer of $\mathcal{S}$ in
$\mathcal{P}_N$ and $\wt(P)$ is the number of qubits on which $P$ acts
nontrivially.

A Calderbank--Shor--Steane (CSS) code is a stabilizer code whose stabilizer
generators can be separated into $X$-type and $Z$-type operators. It is
specified by two binary matrices
$
H_X\in\Ftwo^{m_X\times N},
$
and 
$
H_Z\in\Ftwo^{m_Z\times N},
$
satisfying the commutation condition
$
H_X H_Z^{\mathsf T}=0.
$
Equivalently, stabilizer checks can be represented in binary form as
\[
H_{\rm CSS}
=
\begin{bmatrix}
H_X & 0\\
0   & H_Z
\end{bmatrix}.
\]
The rows of $H_X$ specify $X$-type stabilizers and the rows of $H_Z$
specify $Z$-type stabilizers. If redundant rows are allowed, the number of
encoded logical qubits is
\begin{equation}
\label{eq:css-dimension}
K
=
N-\rank(H_X)-\rank(H_Z).
\end{equation}
Also, for the CSS code, the spaces of
logical equivalence classes are
\begin{equation}
\label{eq:css-logical-spaces}
    \mathcal L_Z
    \triangleq
    \ker H_X/\operatorname{row}(H_Z),
\ \text{and}\ 
    \mathcal L_X
    \triangleq
    \ker H_Z/\operatorname{row}(H_X).
\end{equation}
The inclusions
$\operatorname{row}(H_Z)\subseteq\ker H_X$ and
$\operatorname{row}(H_X)\subseteq\ker H_Z$
follow from the CSS orthogonality condition
$H_XH_Z^T=0$.

An $X$ error is detected/corrected through the $Z$-type checks $H_Z$, while a $Z$
error is detected/corrected through the $X$-type checks $H_X$. The minimum distance of a CSS code can be characterized in terms
of the corresponding binary spaces. The minimum weights of nontrivial
logical $X$- and $Z$-type operators are, respectively,
\begin{align}
D_X
&=
\min\left\{
\wt(x):
x\in\ker(H_Z)\setminus\operatorname{row}(H_X)
\right\},
\label{eq:css-dx}\\
D_Z
&=
\min\left\{
\wt(z):
z\in\ker(H_X)\setminus \operatorname{row}(H_Z)
\right\},
\label{eq:css-dz}
\end{align}
and the quantum minimum distance is
$
D=\min\{D_X,D_Z\}$. Throughout the paper, we use lowercase notation
(e.g., $d_A,d_B, d$) for classical-code minimum distances
and distance thresholds, and uppercase notation
(e.g., $D,D_X,D_Z$) for quantum-code minimum distances and their certified lower bounds.

A family of CSS codes is quantum low-density parity-check (qLDPC) if the
row and column weights of $H_X$ and $H_Z$ remain bounded as the quantum
blocklength grows. In particular, bounded row weight corresponds to
bounded stabilizer weight. 

\subsection{Hypergraph-Product Codes}

Let $H_A\in\Ftwo^{m_A\times n_A}$ and
$H_B\in\Ftwo^{m_B\times n_B}$
be parity-check matrices of two classical binary linear codes. The
hypergraph-product (HGP) construction \cite{TillichZemor2014} defines the
CSS matrices
\begin{align}
H_X &=
\left[
H_A\otimes I_{n_B}
\ \middle|\
I_{m_A}\otimes H_B^{\mathsf T}
\right],\\
H_Z &=
\left[
I_{n_A}\otimes H_B
\ \middle|\
H_A^{\mathsf T}\otimes I_{m_B}
\right],
\end{align}
where $\otimes$ denotes the Kronecker product. These matrices satisfy
$H_XH_Z^{\mathsf T}=0$ by construction, and the quantum blocklength is
\[
N=n_An_B+m_Am_B.
\]

In this work, we use full-row-rank parity-check matrices. Writing
$k_A=n_A-m_A$ and $k_B=n_B-m_B$,
and letting $d_A$ and $d_B$ denote the corresponding classical minimum
distances, the resulting HGP code has \cite{TillichZemor2014}
\[
K=k_Ak_B,
\ \text{and}\ 
D=\min\{d_A,d_B\}.
\]
Let $w_r(H)$ and $w_c(H)$ denote the maximum row
and column weights of a matrix $H$, respectively. Then the maximum weights
of the resulting $X$- and $Z$-type stabilizer generators satisfy
\[
w_X\le w_r(H_A)+w_c(H_B),
\ \text{and}\ 
w_Z\le w_c(H_A)+w_r(H_B).
\]

We will frequently use the \textit{symmetric} case $H_A=H_B=H$, where
$H$ is a full-row-rank parity-check matrix of an $[n,k,d]$ classical code.
In this case,
\[
N=n^2+(n-k)^2,\qquad
K=k^2,\qquad
D=d,
\]
and the maximum stabilizer weight is upper bounded by
$w_r(H)+w_c(H)$.

\subsection{Fair-Density Parity-Check Codes}

Fair-density parity-check (FDPC) codes were introduced by the author in
\cite{mahdavifar2024high} as structured high-rate classical codes with constant
column weight and moderately growing row weight. Let $n=q^2$, where $q\geq 2$ is an integer. The original
FDPC construction in \cite{mahdavifar2024high} was presented
using $q=2t$, and hence even $q$. The same construction extends to arbitrary $q\geq2$. The base FDPC parity-check matrix
$H_b\in\Ftwo^{2q\times q^2}$
has column weight $2$ and row weight $q$. Its columns consist of all
weight-$2$ binary vectors of length $2q$ whose two nonzero positions
have indices differing by an odd number. The matrix has
$\rank(H_b)=2q-1$, and the corresponding base code has minimum distance $4$.

An order-$s$ FDPC code is obtained by stacking $s$ independently
permuted copies of the base matrix,
\begin{equation}
\label{eq:fdpc-order-s}
H_s=
\begin{bmatrix}
H_b\\
\pi_1(H_b)\\
\vdots\\
\pi_{s-1}(H_b)
\end{bmatrix},
\end{equation}
where each $\pi_i$ denotes a permutation of the $n$ columns. Thus,
$H_s$ has $2sq$ rows, row weight $q$, and column weight $2s$. Within each layer, the sums of the checks associated with the
two sides of $K_{q,q}$ are identical, giving one dependency per
layer. Moreover, the corresponding all-ones sums are identical
across the $s$ layers. Hence $H_s$ has at least $2s-1$
independent row dependencies, and
\[
    k\ge n-2sq+2s-1.
\]

The main advantages of FDPC codes for our purposes are their explicit
high-rate structure, controlled row/column weight, analytically tractable
minimum-distance behavior, and a weight distribution that can be
explicitly characterized. Subsequent work has further developed the
classical FDPC framework. In particular,
\cite{Moradi2025HighRateFDPC} introduced a low-complexity encoding
algorithm. Furthermore, \cite{Hosseinzadeh2025FDPCDecoding} developed a layered
normalized min-sum decoder for
improved finite-length decoding performance. These follow-up works further demonstrate
the suitability of FDPC codes as structured classical components for the
quantum constructions considered in this paper.


\section{Extending FDPC Codes: Weight Distribution Analysis and Sparsification}
\label{sec:fdpc-extension}

In this section, we develop two extensions of the classical FDPC framework that make it
particularly suitable for the quantum constructions considered in this work.
First, we study higher-order FDPC codes and characterize their weight distribution
in terms of the FDPC order. Second, we introduce a structured
check-splitting operation for \textit{sparsifying} FDPC codes that reduces the row weight of their parity-check
matrix while keeping their column weight controlled. This is done in such a way that the resulting code, which may itself be an LDPC code, inherits the minimum-distance guarantee of the \textit{parent} FDPC code. For the special case of structured order-$2$ sparsification developed
later in this section, the post-sparsification weight distribution is also available
explicitly. Together, these two ingredients provide the classical matrices needed to obtain both bounded-stabilizer-weight
qLDPC codes and the high-rate CSS constructions developed in the next section.

\subsection{Higher-Order FDPC Codes: Weight Distribution and Minimum Distance Analysis}
\label{subsec:higher-order-fdpc}

We first extend the weight-distribution analysis of FDPC codes, which was developed with a particular focus on order 2 in \cite{mahdavifar2024high}. In contrast to the analysis in \cite{mahdavifar2024high}, which relied
on enumerating \textit{irreducible} codewords in the base code to obtain bounds on the weight
distribution, the analysis here uses the graph-theoretic representation of
the FDPC base code together with the MacWilliams identity. This leads to an
exact characterization of the base-code weight enumerator and, consequently,
an exact expression for the ensemble-average weight distribution of FDPC
codes of arbitrary order.

Recall
that the base matrix $H_b$ can be viewed, up to a permutation of its rows and
columns, as the vertex--edge incidence matrix of the complete bipartite graph with $q$ nodes in each part, denoted by
$K_{q,q}$. Hence, the base FDPC code
$\mathcal{C}_b=\ker(H_b)$ is the cycle space of $K_{q,q}$.

Let $W_b(z)=\sum_{w=0}^{q^2} A_w z^w$
denote the weight enumerator of $\mathcal{C}_b$. The next theorem shows that the bipartite structure of
$H_b$ gives an explicit expression for $W_b(z)$.

\begin{theorem}
\label{thm:fdpc-base-weight}
Let $\mathcal{C}_b=\ker(H_b)$ be the length-$n=q^2$ base FDPC code. Then
\begin{equation}
\label{eq:fdpc-base-we}
W_b(z)
=
2^{-2q}
\sum_{a=0}^{q}\sum_{b=0}^{q}
\binom{q}{a}\binom{q}{b}
(1+z)^{q^2-\delta(a,b)}
(1-z)^{\delta(a,b)},
\end{equation}
where
\begin{equation}
\label{eq:fdpc-delta}
\delta(a,b)=q(a+b)-2ab.
\end{equation}
Equivalently, the number of weight-$w$ codewords is
\begin{equation}
\label{eq:fdpc-base-Aw}
A_w
=
2^{-2q}
\sum_{a=0}^{q}\sum_{b=0}^{q}
\binom{q}{a}\binom{q}{b}
K_w\!\left(\delta(a,b);q^2\right),
\end{equation}
where
\[
K_w(x;n)
=
\sum_{j=0}^{w}
(-1)^j
\binom{x}{j}\binom{n-x}{w-j}
\]
is the binary Krawtchouk polynomial.
\end{theorem}

\begin{proof}
As noted above, $H_b$ is the vertex--edge incidence matrix of
$K_{q,q}$. Therefore, $\mathcal{C}_b$ is the cycle space of $K_{q,q}$,
while its dual $\mathcal{C}_b^\perp=\operatorname{row}(H_b)$ is the cut
space of $K_{q,q}$.

Consider a subset $S$ of the $2q$ vertices containing $a$ vertices from
the left part and $b$ vertices from the right part. The corresponding cut
contains
\[
a(q-b)+(q-a)b
=
q(a+b)-2ab
=
\delta(a,b)
\]
edges. There are $\binom{q}{a}\binom{q}{b}$ choices of such a subset.
Since $S$ and its complement define the same cut, every cut is counted
exactly twice. Hence, the weight enumerator of the cut space is
\begin{equation}
\label{eq:fdpc-cut-we}
W_{\mathcal{C}_b^\perp}(z)
=
\frac{1}{2}
\sum_{a=0}^{q}\sum_{b=0}^{q}
\binom{q}{a}\binom{q}{b}
z^{\delta(a,b)}.
\end{equation}

Since $\rank(H_b)=2q-1$, we have
$|\mathcal{C}_b^\perp|=2^{2q-1}$. Applying the MacWilliams identity to
$\mathcal{C}_b^\perp$ gives
\begin{align}
W_b(z)
&=
\frac{1}{|\mathcal{C}_b^\perp|}
(1+z)^{q^2}
W_{\mathcal{C}_b^\perp}
\left(\frac{1-z}{1+z}\right) \nonumber\\
&=
2^{-2q}
\sum_{a=0}^{q}\sum_{b=0}^{q}
\binom{q}{a}\binom{q}{b}
(1+z)^{q^2-\delta(a,b)}
(1-z)^{\delta(a,b)},
\end{align}
which proves~\eqref{eq:fdpc-base-we}. Extracting the coefficient of $z^w$
and using
\[
[z^w](1+z)^{n-x}(1-z)^x=K_w(x;n)
\]
gives~\eqref{eq:fdpc-base-Aw}.
\end{proof}

We next recall the random ensemble of FDPC codes of a given order $s>1$ used throughout
this work \cite{mahdavifar2024high}. First, the base code
$\mathcal{C}_b$ is fixed and then $s-1$
permutations $\pi_1,\ldots,\pi_{s-1}$ of the $n=q^2$ coordinates are chosen independently and uniformly at random. The
resulting random order-$s$ FDPC code can also be represented by
\begin{equation}
\label{eq:random-order-s-fdpc}
\mathcal{C}^{(s)}
=
\mathcal{C}_b
\cap \pi_1(\mathcal{C}_b)
\cap\cdots\cap\pi_{s-1}(\mathcal{C}_b).
\end{equation}
All expectations in the rest of this section are taken over this random ensemble of order-$s$ FDPC codes. 

\begin{corollary}
\label{cor:fdpc-order-s-weight}
For an order-$s$ FDPC code drawn from the above random ensemble, let
$\cA_w^{(s)}$ denote the random variable representing the number of
codewords of Hamming weight $w$ in that code. Then
\begin{equation}
\label{eq:fdpc-order-s-we}
\mathbb{E}\!\left[\cA_w^{(s)}\right]
=
\frac{A_w^s}{\binom{n}{w}^{s-1}},
\end{equation}
where $A_w$ is given exactly by
Theorem~\ref{thm:fdpc-base-weight}.
\end{corollary}
\begin{proof}
Fix a weight-$w$ support $S$ of the base code. Under a uniformly random
coordinate permutation, the probability that $S$ is also a support of a
weight-$w$ codeword in a permuted copy is
$A_w/\binom{n}{w}$. Since the $s-1$ permutations are independent and the
first layer contains $A_w$ such supports, we have
\[
\mathbb{E}\!\left[\cA_w^{(s)}\right]
=
A_w\left(\frac{A_w}{\binom{n}{w}}\right)^{s-1}.
\]
\end{proof}

For the minimum-distance analysis, we will use the following simple uniform
bound on the base-code weight distribution.

\begin{lemma}
\label{lem:fdpc-base-we-bound}
For the base FDPC code $\cC_b$, we have
$A_w\le q^w$ for $0\le w\le n=q^2$. Moreover, for the random ensemble of order-$s$ FDPC codes we have
\begin{equation}
\label{eq:fdpc-order-s-we-bound}
\mathbb{E}\!\left[\cA_w^{(s)}\right]
\le
\left(
\frac{w^{s-1}}{q^{s-2}}
\right)^w.
\end{equation}
\end{lemma}

\begin{proof}
Since $\mathcal{C}_b$ is the cycle space of $K_{q,q}$, every codeword
corresponds to an edge set having even degree at every vertex. Hence,
$A_w=0$ for odd $w$. Let $w$ be even. For either side of the bipartition,
the $w$ edge endpoints can be paired in
$\frac{w!}{2^{w/2}(w/2)!}$
ways, and the resulting pairs can be assigned to the $q$ vertices in
$q^{w/2}$ ways. Applying this independently to both sides and dividing by
the $w!$ orderings of the $w$ distinct edges gives
\[
A_w
\le
\frac{w!}{2^w((w/2)!)^2}q^w
=
\frac{\binom{w}{w/2}}{2^w}q^w
\le q^w.
\]
For the second claim, Corollary~\ref{cor:fdpc-order-s-weight} and
$\binom{n}{w}\ge(n/w)^w$, with $n=q^2$, give
\[
\mathbb{E}\!\left[\cA_w^{(s)}\right]
\le
\frac{q^{sw}}{(q^2/w)^{(s-1)w}}
=
\left(\frac{w^{s-1}}{q^{s-2}}\right)^w.
\]
\end{proof}

The ensemble-average bound in \eqref{eq:fdpc-order-s-we-bound} can be converted, through a first-moment argument, into a probabilistic lower bound on the minimum distance. This is done in the next theorem.

\begin{theorem}
\label{thm:fdpc-distance}
For every integer $s\ge 3$ and any $0<\delta<1$, a code $\cC$
drawn from the random ensemble of order-$s$ FDPC codes satisfies
\begin{equation}
\label{eq:fdpc-distance-prob}
\Pr\!\left(
d_{\min}(\cC)\ge d_s(\delta)
\right)
\ge 1-\delta,
\end{equation}
where
\begin{equation}
\label{eq:fdpc-Ds}
d_s(\delta)
\,\deff\,
\left\lfloor
\left(\frac{\delta}{1+\delta}\right)^{\frac{1}{s-1}}
q^{\frac{s-2}{s-1}}
\right\rfloor .
\end{equation}
Consequently, for a constant $s\ge3$, there exist order-$s$ FDPC codes
with
$
d_{\min}
=
\Omega\!\left(
q^{\frac{s-2}{s-1}}
\right)
=
\Omega\!\left(
n^{\frac{s-2}{2(s-1)}}
\right)$, where $n=q^2$.
\end{theorem}

\begin{proof}
Let $\rho=\delta/(1+\delta)$. For every $1\le w<d_s(\delta)$,
\eqref{eq:fdpc-order-s-we-bound} gives
$\mathbb{E}\!\left[\cA_w^{(s)}\right]\le \rho^w$.
Therefore,
\[
\mathbb{E}\!\left[
\sum_{w=1}^{d_s(\delta)-1}\cA_w^{(s)}
\right]
\le
\sum_{w=1}^{\infty}\rho^w
=
\frac{\rho}{1-\rho}
=
\delta.
\]
By Markov's inequality, the probability that there exists a nonzero
codeword of weight below $d_s(\delta)$ in $\cC$ is at most $\delta$. This completes the proof of the first part.

For the second part, for any fixed $\delta\in(0,1)$, this probability is strictly
less than one and, hence, there exists a realization satisfying
$d_{\min}\geq d_s(\delta)$. For fixed $s\geq3$, the factor
$(\delta/(1+\delta))^{1/(s-1)}$ is a positive constant, and we have
\[
d_s(\delta)=\Omega\!\left(q^{\frac{s-2}{s-1}}\right)
=\Omega\!\left(n^{\frac{s-2}{2(s-1)}}\right),
\]
where $n=q^2$. This completes the proof.
\end{proof}

Combining Theorem~\ref{thm:fdpc-distance} with the rate bound
\[
R\ge 1-\frac{2s}{q}+\frac{2s-1}{q^2}
\]
provides the rate--distance scaling of higher-order FDPC codes. For a constant
$s\ge3$ and any $0<\delta<1$, the rate satisfies $R=1-O(1/q)$. Also,
with probability at least $1-\delta$,
\[
d_{\min}\ge
\left\lfloor
\left(\frac{\delta}{1+\delta}\right)^{\frac{1}{s-1}}
q^{\frac{s-2}{s-1}}
\right\rfloor.
\]
If the order is allowed to grow as $s=\Theta(\log q)$, the rate remains
\[
R=1-O\!\left(\frac{\log q}{q}\right),
\]
while the minimum distance scales as $d_{\min}=\Omega(q)=\Omega(\sqrt n)$.

Note that the case $s=2$ requires a separate analysis. Corollary~\ref{cor:fdpc-order-s-weight}
remains valid for $s=2$ and gives the exact ensemble average.
However, the uniform bound of Lemma~\ref{lem:fdpc-base-we-bound} reduces in
this case to
\[
\mathbb{E}\!\left[\cA_w^{(2)}\right]\le w^w,
\]
which does not yield a growing minimum-distance guarantee through the above
first-moment argument. This is why Theorem~\ref{thm:fdpc-distance} is stated
for $s\ge3$. The order-$2$ case was studied in more detail in
\cite{mahdavifar2024high}.

\subsection{Sparsification of FDPC Codes}
\label{subsec:check-splitting}

The parity-check matrix of an order-$s$ FDPC code has column weight $2s$,
which can be fixed to an even constant, whereas its row weight is
$q=\sqrt{n}$ and, hence, grows with the blocklength. This growth is closely
tied to the high-rate nature of the construction. Note that since $H_s$ has
$2sq$ rows and $n=q^2$ columns, its check-to-variable ratio is
\[
\frac{2sq}{n}=\frac{2s}{q},
\]
which vanishes as $q$ grows for fixed $s$. Thus, with a fixed column
weight, allowing the row weight to grow is what enables the rate of the
FDPC family to approach one. In contrast, keeping both the row and column
weights bounded leads naturally to a nonvanishing check density and, hence,
to a rate bounded away from one.

This observation suggests using a high-rate FDPC code as a \emph{parent}
code and trading some of its rate for sparser parity checks.
The key question is whether the high-weight checks of the parent code can
be replaced by several lower-weight checks without sacrificing its useful
distance properties. We show that this can be done through a simple
check-splitting operation. Each weight-$q$ row is partitioned into
$\ell=q/L$ disjoint checks of weight $L$, for some integer $L>1$ that divides $q$, while every variable remains in
exactly the same number, $2s$, of checks. Consequently, the resulting
matrix has row weight $L$ and column weight $2s$, while its code is a
subcode of the parent FDPC code and, therefore, inherits its
minimum-distance guarantee.

This operation provides an explicit interpolation between high-rate FDPC
codes and conventional regular LDPC codes. By choosing $L$ to grow with $q$, one
can retain a rate approaching one while reducing the row density of the
parent matrix. On the other hand, for fixed $s$ and fixed $L>2s$, the
resulting matrices have bounded row and column weights. Thus, check splitting with an underlying fixed parameter $L$ turns the FDPC construction into
a family of $(2s,L)$-regular LDPC codes.

Next, we specify this check-splitting operation in more detail and formally refer to it as \textit{sparsification}. 

\begin{definition}[Sparsification of an FDPC matrix]
\label{def:check-splitting}
Let $H_s$ be an order-$s$ FDPC parity-check matrix of length $n=q^2$,
with rows $h_1,\ldots,h_m$, where $m=2sq$. Let $L$ divide $q$ and set
$\ell=q/L$. For each $i\in[m]$, partition the support of $h_i$ as
\[
\supp(h_i)
=
S_i^{(1)}\mathbin{\dot\cup}S_i^{(2)}
\mathbin{\dot\cup}\cdots\mathbin{\dot\cup}S_i^{(\ell)},
\ \text{with}\ 
|S_i^{(j)}|=L,
\]
where $\dot\cup$ denotes disjoint union. Let $h_i^{(j)}\in \Ftwo^n$ be the
indicator vector of $S_i^{(j)}$. Thus,
$
h_i=\sum_{j=1}^{\ell}h_i^{(j)}.
$
The \emph{$L$-sparsification} of $H_s$, denoted by
$\widetilde H_s$, is the parity-check matrix whose rows are $h_i^{(j)}$, for $i\in[m], j\in[\ell].$
The corresponding code
\[
\widetilde{\cC}_s=\ker(\widetilde H_s)
\]
is called the $L$-sparsified FDPC code.
\end{definition}

\begin{lemma}
\label{lem:sparsification-distance}
Let $\widetilde{\cC}_s=\ker(\widetilde H_s)$ of rate $\widetilde R$ be an $L$-sparsified FDPC code corresponding to the parent code $\cC_s$. 
Then
\[
\widetilde{\cC}_s\subseteq\cC_s,
\ \text{and}\ 
d_{\min}(\widetilde{\cC}_s)
\ge
d_{\min}(\cC_s).
\]
Moreover,
\begin{equation}
    \dim(\widetilde{\cC}_s)
    \ge
    n-\frac{2sn}{L}+2s-1,
\end{equation}
and consequently
\begin{equation}
    \widetilde R
    \ge
    1-\frac{2s}{L}+\frac{2s-1}{n}.
\end{equation}
\end{lemma}

\begin{proof}
Each row of $H_s$ is the sum of the rows replacing it in $\widetilde H_s$. Hence, $\operatorname{row}(H_s)\subseteq\operatorname{row}(\widetilde H_s)$,
which implies $\widetilde{\cC}_s\subseteq\cC_s$. This shows the first part. For the second part, $\widetilde H_s$ contains $2sn/L$ rows.
Within each FDPC layer, the sum of all sparsified checks associated
with one side of $K_{q,q}$ equals the all-ones vector, and the same
holds for the opposite side. Hence each layer contributes one linear
dependency. Moreover, these all-ones sums are identical across the
$s$ layers, yielding $s-1$ additional independent dependencies.
Therefore,
\[
    \rank(\widetilde H_s)
    \le
    \frac{2sn}{L}-(2s-1),
\]
which gives
\[
    \dim(\widetilde{\cC}_s)
    \ge
    n-\frac{2sn}{L}+2s-1.
\]
Dividing by $n$ gives the stated rate bound.
\end{proof}

Beyond the deterministic distance guarantee above, sparsification can
substantially suppress the low-weight codewords of the parent FDPC code.
To quantify this effect, we consider random $L$-sparsification, in which
the support of each parent check is partitioned independently and uniformly
into $\ell=q/L$ subsets of size $L$.

Suppose a vector overlaps a given weight-$q$ parent check in $d$ positions.
The probability that it satisfies all $\ell$ checks produced by the
sparsification is
\begin{equation}
\label{eq:split-survival}
p_{q,L}(d)
=
\frac{[z^d]\,\Phi_L(z)^\ell}{\binom{q}{d}},
\ \text{where}\ 
\Phi_L(z)
=
\sum_{\substack{j=0\\ j\ {\rm even}}}^{L}
\binom{L}{j}z^j .
\end{equation}
For example,
$p_{q,L}(0)=1$, while $p_{q,L}(2)=\frac{L-1}{q-1}$.
This illustrates the additional suppression of low-weight configurations.

For a fixed parent matrix $H_s$ with rows $h_1,\ldots,h_m$, let
$d_i(x)=|\supp(x)\cap\supp(h_i)|$. If $\widetilde{\cA}_w$ denotes the
number of weight-$w$ codewords after random sparsification, using independence
of the row partitions we get
\begin{equation}
\label{eq:split-weight-average}
\mathbb{E}_{\rm sp}\!\left[\widetilde{\cA}_w\mid H_s\right]
=
\sum_{\substack{x\in\cC_s\\ \wt(x)=w}}
\prod_{i=1}^{m} p_{q,L}\!\left(d_i(x)\right),
\end{equation}
where the expectation $\mathbb{E}_{\rm sp}\!\left[. \mid H_s \right]$ is taken over the independent random
$L$-sparsifications of the rows of a given $H_s$, conditioned on the parent matrix
$H_s$ being fixed. Consequently, in the same probability space, a first-moment bound yields
\begin{equation}
\label{eq:split-distance-first-moment}
\Pr\!\left(
d_{\min}(\widetilde{\cC}_s)<d\mid H_s
\right)
\le
\sum_{w=1}^{d-1}
\mathbb{E}_{\rm sp}\!\left[\widetilde{\cA}_w\mid H_s\right].
\end{equation}
Equations~\eqref{eq:split-weight-average} and
\eqref{eq:split-distance-first-moment} apply to general random
sparsification, but their evaluation requires the detailed overlap
profiles of the parent codewords. For the finite-length constructions
considered later, we next identify an order-$2$ specialization for which
the post-sparsification weight distribution is available explicitly.

\subsection{A Special Case: Order-$2$ Sparsified FDPC Codes}
\label{subsec:order2-sparsified}

The general sparsification analysis above applies to arbitrary choices of
the partitions in Definition~\ref{def:check-splitting}. In particular,
Eq.~\eqref{eq:split-weight-average} gives the average weight distribution
under random sparsification. Its direct evaluation, however, requires the
overlap profile of each parent-codeword support with the individual
parity checks. For the finite-length constructions of
Section~\ref{sec:finite-qfdpc}, it is useful
to consider a structured order-$2$ specialization for which both
the post-sparsification weight distribution and a stronger
dimension bound are readily available.

Let $q=\ell L$. Recall that a base FDPC layer is the cycle code of
$K_{q,q}$. Partition each of the two vertex sets of $K_{q,q}$ into
$\ell$ groups of size $L$, and, for every vertex, split its weight-$q$
star check according to the $\ell$ groups on the opposite side. This is
a particular $L$-sparsification in the sense of
Definition~\ref{def:check-splitting}. Under this splitting, the $q^2$
coordinates of one layer decompose into $\ell^2$ disjoint blocks, each
corresponding to a copy of $K_{L,L}$.

Let $\cC_{q,L}^{\rm sp}$ denote the resulting one-layer code. Then
\begin{equation}
    \cC_{q,L}^{\rm sp}
    \cong
    \cC_b(L)^{\oplus \ell^2},
    \label{eq:order2-sparsified-direct-sum}
\end{equation}
where $\oplus$ denotes the direct sum of linear codes and
$\cC_b(L)^{\oplus \ell^2}$ consists of $\ell^2$ independent copies of
the base FDPC code $\cC_b(L)$ with parameter $L$ acting on disjoint sets of $L^2$
coordinates. Accordingly, if $W_L(z)$ denotes the weight enumerator in
Theorem~\ref{thm:fdpc-base-weight} evaluated with $q=L$, then
\begin{equation}
    B_{q,L}(z)
    \triangleq
    W_L(z)^{\ell^2}
    =
    \sum_{w=0}^{q^2} B_{q,L,w} z^w
    \label{eq:order2-sparsified-layer-enumerator}
\end{equation}
is the exact weight enumerator of one sparsified layer. 
Moreover, let $k_0$ denote the dimension of this one-layer
sparsified code. Since each of the $\ell^2$ disjoint copies of
$\mathcal C_b(L)$ has dimension $(L-1)^2$, we have
\begin{equation}
k_0 \triangleq \dim(\mathcal C^{\rm sp}_{q,L})
= \ell^2(L-1)^2.
    \label{eq:order2-sparsified-layer-dimension}
\end{equation}

For the order-$2$ ensemble, let the relative coordinate permutation
between the two sparsified layers be uniformly random, and define
\begin{equation}
    \widetilde{\cC}_{q,L}^{(2)}
    =
    \cC_{q,L}^{\rm sp}
    \cap
    \pi\!\left(\cC_{q,L}^{\rm sp}\right).
    \label{eq:order2-sparsified-ensemble}
\end{equation}

\begin{proposition}
\label{prop:order2-sparsified}
The parity-check matrix of
$\widetilde{\cC}_{q,L}^{(2)}$ has row weight $L$ and column weight $4$.
If $\widetilde{\cA}_w$ denotes the number of weight-$w$ codewords, then
\begin{equation}
    \mathbb E[\widetilde{\cA}_w]
    =
    \frac{B_{q,L,w}^2}{\binom{q^2}{w}}.
    \label{eq:order2-sparsified-wd}
\end{equation}
Furthermore,
\begin{equation}
    \dim\!\left(\widetilde{\cC}_{q,L}^{(2)}\right)
    \ge
    2\ell^2(L-1)^2-q^2
    =
    \ell^2(L^2-4L+2).
    \label{eq:order2-sparsified-dimension}
\end{equation}
\end{proposition}

\begin{proof}
The row and column weights follow directly from the order-$2$
sparsification. For the weight distribution, a fixed weight-$w$
codeword of the first sparsified layer belongs to the uniformly permuted
second layer with probability
$B_{q,L,w}/\binom{q^2}{w}$. Hence, by the same permutation-counting
argument used in Corollary~\ref{cor:fdpc-order-s-weight},
\eqref{eq:order2-sparsified-wd} follows. Finally, the intersection of
two $k_0$-dimensional subspaces of $\Ftwo^{q^2}$ has dimension at least
$2k_0-q^2$, which together with
\eqref{eq:order2-sparsified-layer-dimension} gives
\eqref{eq:order2-sparsified-dimension}.
\end{proof}

The structured order-$2$ ensemble of Proposition~\ref{prop:order2-sparsified} will be used
throughout the finite-length constructions in Section~VI. Next, we illustrate its dimension bound and the resulting
check-weight reduction for some examples.

\noindent
{\bf Example 1.}
Consider first the structured order-$2$ sparsified FDPC ensemble with
$q=15$ and $L=5$. Here, $n=q^2=225$ and $\ell=q/L=3$.
For comparison, the unsparsified order-$2$ FDPC code has row weight
$15$, column weight $4$, and dimension
\[
    k\geq n-4q+3
    =225-60+3
    =168.
\]
Under the structured sparsification of this subsection, the row weight
is reduced to $L=5$ while the column weight remains equal to $4$.
Proposition~\ref{prop:order2-sparsified} gives
\[
    \widetilde{k}
    \geq
    \ell^2(L^2-4L+2)
    =
    3^2(25-20+2)
    =
    63.
\]
Thus, the structured sparsification changes the classical parameters
from a high-rate FDPC code with rate at least $168/225$ to a
$(4,5)$-regular LDPC code with rate at least $63/225$.

As a second example, let $q=32$ and $L=8$, so that
$n=1024$ and $\ell=4$. The unsparsified order-$2$ FDPC code has
row weight $32$, column weight $4$, and dimension
\[
    k\geq n-4q+3
    =1024-128+3
    =899.
\]
The structured sparsification gives a $(4,8)$-regular parity-check
matrix, and Proposition~\ref{prop:order2-sparsified} yields
\[
    \widetilde{k}
    \geq
    \ell^2(L^2-4L+2)
    =
    4^2(64-32+2)
    =
    544.
\]
Hence, the guaranteed rate is reduced from at least $899/1024$ to
at least $544/1024$, while the row weight is reduced from $32$ to $8$.

\section{Quantum CSS Codes from FDPC and Sparsified FDPC Codes}
\label{sec:qfdpc}

We now use FDPC and sparsified FDPC parity-check matrices as classical
components of the hypergraph-product construction. The unsparsified FDPC
matrices provide the high-rate classical structure and distance guarantees
developed in the previous section, while sparsification allows the row
weight to be reduced in a controlled manner. This leads to two useful
quantum regimes: qLDPC codes with bounded stabilizer weight, obtained from
fixed-order sparsified FDPC matrices, and high-rate CSS codes obtained by
allowing the FDPC order and sparsified row weight to scale moderately with
blocklength.

\subsection{General Hypergraph-Product Construction}
\label{subsec:qfdpc-general}

We first consider an asymmetric HGP construction in which the
FDPC-based component is paired with an arbitrary second classical
code. This flexibility is useful
for controlling the resulting quantum dimension and distance, and includes
the symmetric construction as an important special case.

Let $H_A$ be obtained from an $L$-sparsified order-$s$ FDPC matrix
$\widetilde H_s$ of length $n=q^2$ by removing linearly dependent rows.
The original FDPC matrix is included as the special case $L=q$. Denote
\[
r_A=\rank(\widetilde H_s),\ \ k_A=n-r_A,
\]
and let $d_A$ be the minimum distance of the classical code $\ker(H_A)=\ker(\widetilde H_s)$. Note that removing \textit{redundant} rows does not change
the classical code and cannot increase the row or column weights. Hence, we have
\begin{equation}
\label{eq:fdpc-component-bounds}
r_A\le \frac{2sn}{L},\ \ 
w_r(H_A)\le L,\ \text{and}\ 
w_c(H_A)\le 2s,
\end{equation}
while Lemma~\ref{lem:sparsification-distance} gives
$d_A\ge d_{\min}(\cC_s)$.

As the second component, let $H_B$ be any full-row-rank parity-check
matrix of a classical $[n_B,k_B,d_B]$ code, and write
$r_B=n_B-k_B$.

Applying the hypergraph-product construction to $H_A$ and $H_B$ results in a CSS code with parameters specified in the following theorem. 

\begin{theorem}
\label{thm:qfdpc-general}
The hypergraph-product construction based on $H_A$ and $H_B$ gives a $[[N,K,D]]$ CSS
code with
\begin{equation}
\label{eq:qfdpc-general-parameters}
N=nn_B+r_Ar_B,\ \ 
K=k_Ak_B,\ \text{and}\ 
D=\min\{d_A,d_B\}.
\end{equation}
Its $X$- and $Z$-stabilizer weights satisfy
\begin{equation}
\label{eq:qfdpc-general-weights}
w_X\le L+w_c(H_B),
\ \text{and}\ 
w_Z\le 2s+w_r(H_B).
\end{equation}
Furthermore, we have
\begin{align}
N
&\le
nn_B\left(
1+\frac{2s}{L}\frac{r_B}{n_B}
\right),
\label{eq:qfdpc-general-N-bound}\\
K
&\ge
nk_B\left(1-\frac{2s}{L}\right).
\label{eq:qfdpc-general-K-bound}
\end{align}
\end{theorem}

\begin{proof}
Since both $H_A$ and $H_B$ are full row rank, the standard
hypergraph-product formulas give \eqref{eq:qfdpc-general-parameters}.
The stabilizer-weight bounds follow from
\eqref{eq:fdpc-component-bounds}, while
\eqref{eq:qfdpc-general-N-bound} and
\eqref{eq:qfdpc-general-K-bound} follow from
$r_A\le2sn/L$.
\end{proof}

Let
\[
R_A=\frac{k_A}{n},
\ \text{and}\ 
R_B=\frac{k_B}{n_B}
\]
denote the rates of the two classical components. Since
$r_A/n=1-R_A$ and $r_B/n_B=1-R_B$, the quantum rate, denoted by $R_Q$, is given by
\begin{equation}
\label{eq:qfdpc-general-rate}
R_Q
=
\frac{R_AR_B}
{1+(1-R_A)(1-R_B)}.
\end{equation}
In particular, using \eqref{eq:fdpc-component-bounds},
\begin{equation}
\label{eq:qfdpc-general-rate-bound}
R_Q
\ge
\frac{
\left(1-\frac{2s}{L}\right)R_B
}{
1+\frac{2s}{L}(1-R_B)
}.
\end{equation}
Thus, the second classical component provides an additional degree of
freedom in trading off quantum rate and distance, while the FDPC component
controls the structured sparsity and supplies the inherited distance
guarantee.

\subsection{Quantum FDPC Codes}
\label{subsec:qfdpc}

A particularly natural specialization of the general HGP
construction is to use the same FDPC component for both
classical codes, i.e., $H_B=H_A$. In this symmetric case,
let $r=r_A, k=k_A, d=d_A$. Then
the quantum parameters reduce to
\begin{equation}
\label{eq:qfdpc-symmetric-parameters}
N=n^2+r^2, \ \
K=k^2,\ \text{and}\ 
D=d,
\end{equation}
with stabilizer weights of the resulting CSS code upper bounded by $L+2s$.
For $L\geq 2s$, the dimension bound in Eq.\,\eqref{eq:qfdpc-general-K-bound} can be
squared in the symmetric construction, giving
\[
N \leq n^2\left(1+\frac{4s^2}{L^2}\right),
\qquad
K \geq n^2\left(1-\frac{2s}{L}\right)^2,
\]
and therefore,
\begin{equation}
\label{eq:qfdpc-symmetric-rate-bound}
R_Q = \frac{R_A^2}{1+(1-R_A)^2}
\ge
\frac{
\left(1-\frac{2s}{L}\right)^2
}{
1+\left(\frac{2s}{L}\right)^2
}.
\end{equation}

This construction gives rise to the quantum FDPC codes defined formally next. 

\begin{definition}[Quantum FDPC codes]
\label{def:qfdpc}
A \emph{quantum FDPC (qFDPC) code} is a CSS code obtained through the
symmetric hypergraph-product construction using an $L$-sparsified
order-$s$ FDPC code as both classical components. The original
unsparsified FDPC construction is included as the special case $L=q$.
\end{definition}

The parameters $s$ and $L$ provide two distinct degrees of freedom in the
qFDPC construction. The order $s$ governs the minimum-distance behavior of
the parent FDPC ensemble, whereas $L$ controls the rate loss introduced by
sparsification and, together with $s$, the stabilizer weight. Later in this section, we show that an appropriate joint scaling of these parameters
produces high-rate CSS codes with polynomially growing minimum
distance and substantially reduced stabilizer weight. A
polylogarithmic stabilizer-weight specialization is given
subsequently in Corollary\,\ref{cor:qfdpc-main-scaling}. 

Throughout the large-blocklength analysis in the remainder of this section, the integer parameters
are taken along sequences satisfying the divisibility condition
$L\mid q$ required by the sparsification construction. 

For the joint scaling mentioned above, we choose $s$ and $L$ jointly as functions of $q$. Let
\[
s=\Theta(\log q),
\ \text{and}\ 
L=s\,g(q),
\]
where $g(q)\to\infty$ and $s g(q)=o(q)$. The first condition is sufficient
to obtain the $\Omega(q)$ minimum-distance scaling of the parent FDPC
ensemble, while the second ensures that sparsification reduces the check
weight substantially below its original value $q$. These choices of parameters lead to the following theorem. 

\begin{theorem}
\label{thm:qfdpc-scaling}
Consider a sequence of integer parameters $(q,s,L)$ satisfying
$L\mid q$, $s=\Theta(\log q)$, and $L=s g(q)$, where $g(q)\to\infty$ and
$s g(q)=o(q)$. For any fixed $0<\delta<1$, the resulting qFDPC code
satisfies, with probability at least $1-\delta$,
\begin{equation}
\label{eq:qfdpc-distance-explicit}
D\ge
\left\lfloor
\left(\frac{\delta}{1+\delta}\right)^{\frac{1}{s-1}}
q^{\frac{s-2}{s-1}}
\right\rfloor .
\end{equation}
Consequently,
\[
D=\Omega(N^{1/4}),
\]
where the implied constant may depend on $\delta$ and on the constants in
$s=\Theta(\log q)$. Moreover,
\[
R_Q=1-O\!\left(\frac{1}{g(q)}\right),
\]
and the stabilizer weight in the resulting qFDPC code is
\[
w_{\rm stab}=O\!\left(\log N\,g(q)\right).
\]
\end{theorem}

\begin{proof}
Since $g(q)\to\infty$, we have $L=sg(q)\geq 2s$ for all
sufficiently large $q$. Hence, Eq.~\eqref{eq:qfdpc-symmetric-rate-bound} applies and, using
$L=sg(q)$, gives
\[
R_Q
\ge
\frac{(1-2/g(q))^2}{1+4/g(q)^2}
=
1-O\!\left(\frac{1}{g(q)}\right).
\]
By Theorem~\ref{thm:fdpc-distance} and the fact that sparsification cannot
decrease minimum distance, \eqref{eq:qfdpc-distance-explicit} holds with
probability at least $1-\delta$.

Since $s=\Theta(\log q)$ and $\delta$ is fixed, we have
\[
\left(\frac{\delta}{1+\delta}\right)^{1/(s-1)}
=\Theta(1),
\]
and also
\[
q^{(s-2)/(s-1)}
=
q\,q^{-1/(s-1)}
=
\Theta(q).
\]
Hence, up to the inconsequential floor operation,
$D=\Omega(q)$.
Furthermore,
\[
n^2\le N\le n^2\left(1+\frac{4}{g(q)^2}\right).
\]
Hence, $N=\Theta(q^4)$ and, therefore, $D=\Omega(N^{1/4})$. Finally,
\[
w_{\rm stab}\le L+2s=s(g(q)+2)
=O\!\left(\log N\,g(q)\right).
\]
\end{proof}

\begin{corollary}
\label{cor:qfdpc-main-scaling}
Consider an integer parameter sequence satisfying $L\mid q$,
$s=\Theta(\log q)$, and
$L=\Theta(\log q\log\log q)$.
Then, for any fixed $0<\delta<1$, the resulting qFDPC code satisfies,
with probability at least $1-\delta$,
\[
D=\Omega(N^{1/4}),
\]
while
\begin{equation}
\label{eq:qfdpc-main-scaling}
R_Q
=
1-O\!\left(\frac{1}{\log\log N}\right),
\end{equation}
and
\[
w_{\rm stab}
=
O(\log N\log\log N).
\]
\end{corollary}

\begin{proof}
Set $g(q)=\Theta(\log\log q)$ in
Theorem~\ref{thm:qfdpc-scaling}. Since $N=\Theta(q^4)$,
we have $\log q=\Theta(\log N)$ and
$\log\log q=\Theta(\log\log N)$, which completes the proof.
\end{proof}

For comparison, setting $L=q$ recovers the original FDPC matrix without
sparsification. In this case, for $s=\Theta(\log q)$,
\[
R_Q
=
1-O\!\left(\frac{\log q}{q}\right)
=
1-O\!\left(\frac{\log N}{N^{1/4}}\right),
\]
while the same probabilistic distance scaling
\[
D=\Omega(N^{1/4})
\]
is retained. However, the stabilizer weight now grows as
\[
w_{\rm stab}
=
q+2s
=
\Theta(N^{1/4}).
\]
Thus, sparsification trades part of the excess rate of the original FDPC
construction for a substantial reduction in stabilizer weight, from
polynomial to polylogarithmic in the quantum blocklength.

\noindent
{\bf Remark 1.} The distance bound in Theorem~\ref{thm:qfdpc-scaling} is conservative in the sense
that it only uses the minimum-distance guarantee of the parent FDPC code.
In particular, the additional parity constraints introduced by
$L$-sparsification are used to control the rate and stabilizer weight, but
their effect on the minimum distance is not exploited. The weight-distribution analysis of Section~\ref{sec:fdpc-extension}
suggests that sparsification can substantially suppress low-weight codewords. Therefore, sparsification
may yield stronger distance scaling than the parent-code
guarantee alone. Establishing an
improved asymptotic distance scaling for the sparsified ensemble is
therefore left for future work.

In addition to the symmetric qFDPC construction, we also
consider an asymmetric HGP construction in which one classical
component is an FDPC code and the other is a repetition code.
This construction is useful for two reasons. First, the repetition
length provides an additional finite-length design parameter for
trading blocklength and distance while retaining bounded stabilizer
weight. Second, as shown in Section~V, the logical classes in one
Pauli sector can be characterized exactly in terms of the codewords
of the FDPC component. This provides an explicit connection between
the FDPC weight distribution and the ML performance of the resulting
quantum code over the erasure channel.

\begin{definition}[Quantum FDPC--Rep codes]
\label{def:qfdpc-rep}
Let $H_A$ be the parity-check matrix of an $L$-sparsified order-$s$ FDPC
code, and let $H_R$ be a full-row-rank parity-check matrix of the
length-$r$ repetition code. Throughout, we take $H_R$ to be the standard path parity-check
matrix of the repetition code, whose rows have weight $2$ and
whose maximum column weight is $2$. A \emph{quantum FDPC--Rep
(qFDPC--Rep) code} is the CSS code obtained from the hypergraph product construction with underlying components
$H_A$ and $H_R$.
\end{definition}

We use the term qFDPC--Rep to distinguish this asymmetric construction
from the symmetric qFDPC codes of Definition~\ref{def:qfdpc}. In
Section~\ref{sec:finite-qfdpc}, we will use both constructions to obtain
explicit finite-length quantum codes with bounded stabilizer weight.

\section{Weight Distribution and ML Performance on the Quantum Erasure Channel}
\label{sec:erasure-distribution}

In this section, we connect the finite-length weight-distribution analysis of FDPC codes to maximum-likelihood (ML) decoding over the quantum erasure channel. We first review the erasure model and show how the classical FDPC weight enumerator can be used to bound the ML probability of error. We then transfer these weight-distribution results to the logical operators of hypergraph-product quantum FDPC codes and derive a first-order approximation to the ML probability of error. 

\subsection{Quantum erasure channel and the role of the weight distribution}
\label{subsec:erasure-distribution}

The quantum erasure channel is one of the standard models for studying the
information-theoretic performance of quantum error-correcting codes
\cite{Bennett1997QuantumErasure,Grassl1997QuantumErasure}. Each physical
qubit in this model is erased independently with probability $p$, and the locations of
the erased qubits are revealed to the decoder. Throughout the rest of this paper,
$p$ denotes the physical erasure probability of the quantum erasure
channel. For
notational simplicity, we use the same symbol $p$ for both the classical
binary erasure channel and the quantum erasure channel, with the channel
being clear from context. Conditioned on an erasure,
the affected qubit can equivalently be regarded as undergoing a uniformly
random Pauli error. Thus, for a CSS code with check matrices $H_X$ and
$H_Z$, decoding over a fixed erased set $E$ reduces to solving
\begin{equation}
    H_{Z,E} x_E = s_X,
\ \text{and}\ 
    H_{X,E} z_E = s_Z,
    \label{eq:qec-erasure-systems}
\end{equation}
where $H_{X,E}$ and $H_{Z,E}$ denote the corresponding erased-column
submatrices.

Erasure decoding of quantum codes has been studied extensively in several
settings. Delfosse and Z\'emor showed that maximum-likelihood (ML) decoding
of surface codes over the quantum erasure channel can be performed in
linear time \cite{DelfosseZemor2020Erasure}.
For hypergraph-product codes, \cite{Connolly2024FastErasure} developed
a fast erasure decoder based on peeling and a vertical--horizontal
decomposition. Also, \cite{Gokduman2024ErasureBPGD} studied belief-propagation
decoding with guided decimation for general qLDPC codes. Recent approaches also include cluster decomposition
\cite{YaoGokdumanPfister2025Cluster} and stabilizer-assisted
inactivation~\cite{Pech2026StabilizerInactivation}, both aimed at approaching or attaining
ML erasure-decoding performance with reduced complexity.

The erasure model is also of increasing practical relevance. In neutral-atom
architectures, dominant error mechanisms can in some cases be detected and
converted into located erasures \cite{Wu2022Erasure}. This possibility has
motivated the study of codes tailored to biased erasure noise
\cite{Sahay2023BiasedErasure}, and mid-circuit conversion of gate errors into
erasures has been demonstrated experimentally in metastable
${}^{171}\mathrm{Yb}$ atomic qubits \cite{Ma2023ErasureConversion}.

The ML decoding performance of a binary linear code
$\mathcal C\subseteq\mathbb F_2^n$ over a classical binary erasure channel
with erasure probability $p$ admits an explicit finite-length bound in
terms of its weight distribution~\cite{RichardsonUrbanke2008}. Let
$W_{\mathcal C}(z)=\sum_{w=0}^{n}A_wz^w$ denote the weight enumerator of
$\mathcal C$. Averaging over an i.i.d.\ erasure pattern, $W_{\mathcal C}(p)$
is precisely the expected number of codewords consistent with the unerased
observations, and straightforward bounding gives
\begin{equation}
    P_{\mathrm e}^{\mathrm{ML}}(p)
    \leq
    \frac{1}{2}
    \left[
        W_{\mathcal C}(p)-1
    \right]
    =
    \frac{1}{2}
    \sum_{w=1}^{n} A_w p^w .
    \label{eq:bec-distribution-bound}
\end{equation}
For the FDPC ensembles considered here, the ensemble-average weight
distribution is available analytically. Hence, averaging
\eqref{eq:bec-distribution-bound} over the ensemble yields a computable
finite-length bound
\begin{equation}
    \mathbb E_{\mathcal C}\!\left[P_{\mathrm e}^{\mathrm{ML}}(p)\right]
    \leq
    \frac{1}{2}
    \sum_{w=1}^{n}
    \mathbb E_{\mathcal C}[\cA_w]\, p^w .
    \label{eq:bec-ensemble-distribution-bound}
\end{equation}

\subsection{Logical weight distribution of quantum FDPC codes}
\label{subsec:logical-distribution}

The logical structure of HGP codes is well understood.
The original construction of Tillich and Z\'emor gives the dimension and
distance of the product code in terms of the two classical component codes
\cite{TillichZemor2014}. More explicitly, HGP logical operators can be
chosen to have support on individual rows or columns of the two product
blocks. It was later shown that these ``line''
operators can be chosen to form a symplectic canonical basis
\cite{Quintavalle2023Logical}. Hence, the codewords of the classical
components appear naturally as representatives of quantum logical
operators.

For the qFDPC--Rep construction of Definition~\ref{def:qfdpc-rep}, we are
interested in a stronger quantity than the weight of a particular logical
representative: the minimum weight within each logical equivalence class.
In fact, two representatives that differ by a stabilizer implement the same
logical operator but may have different Hamming weights. Consider therefore
an HGP code with component matrices $H_A$ and $H_B$, where $H_A$ is the
FDPC component and $H_B=H_R$ is the parity-check matrix of a repetition
code. For this qFDPC--Rep family, the complete minimum-weight logical-class
distribution in one Pauli sector can be characterized exactly. 

\begin{proposition}[Logical-class weight distribution of qFDPC--Rep codes]
\label{prop:fdpc-rep-logical-distribution}
Consider a qFDPC--Rep code of Definition~\ref{def:qfdpc-rep}. Let
$H_A\in\mathbb F_2^{m\times n}$ be a full-row-rank parity-check matrix
of the FDPC component, with
$\mathcal C_A=\ker H_A$,
and let $H_R\in\mathbb F_2^{(r-1)\times r}$ be a full-row-rank
parity-check matrix of the length-$r$ repetition code. With
$\mathcal L_Z$ defined as in Eq.~\eqref{eq:css-logical-spaces}, there is
an isomorphism
\begin{equation}
    \mathcal L_Z \cong \mathcal C_A .
    \label{eq:fdpc-rep-logical-quotient}
\end{equation}
Under this isomorphism, the minimum physical weight of the logical class
associated with $\ba \in\mathcal C_A$ is exactly
$\operatorname{wt}(\ba)$. Therefore, if
\begin{equation}
    W_A(z)=\sum_{w=0}^{n} A_w^{(A)}z^w
\end{equation}
is the weight enumerator of the FDPC component and $M_w^Z$ denotes the
number of $Z$-logical classes having minimum physical weight $w$, then we have
\begin{equation}
        M_w^Z=A_w^{(A)}.
    \label{eq:fdpc-rep-exact-logical-distribution}
\end{equation}
\end{proposition}

\begin{proof}
Represent a $Z$-type operator by a pair $(U,V)$, where
$U\in\mathbb F_2^{n\times r}$, and $V\in\mathbb F_2^{m\times(r-1)}$.
The condition $(U,V)\in\ker H_X$ is equivalent to
\begin{equation}
    H_AU+VH_R=0.
    \label{eq:fdpc-rep-z-cycle}
\end{equation}
Define
\begin{equation}
    \phi_Z:\ker H_X\longrightarrow\mathcal C_A,
\ \text{and}\ 
    \phi_Z(U,V)=U\mathbf{1}_r .
    \label{eq:fdpc-rep-logical-map}
\end{equation}
Since $H_R\mathbf{1}_r=0$, Eq.~\eqref{eq:fdpc-rep-z-cycle} gives
$H_AU\mathbf{1}_r=0$, so $\phi_Z(U,V)\in\mathcal C_A$.

The map $\phi_Z$ is surjective. In fact, for any
$\ba\in\mathcal C_A$ and any $j\in[r]$,
$(\ba\bee_j^T,0)\in\ker H_X$ and
$\phi_Z(\ba\bee_j^T,0)=\ba$.
Moreover, every $Z$-stabilizer can be written as
$(TH_R,H_AT)$, for some $T\in\mathbb F_2^{n\times(r-1)}$,
and hence belongs to $\ker\phi_Z$. Conversely, if
$\phi_Z(U,V)=0$, then every row of $U$ has even parity. Since
$\operatorname{row}(H_R)$ is the even-parity subspace of
$\mathbb F_2^r$, there exists
$T\in\mathbb F_2^{n\times(r-1)}$ such that $U=TH_R$. Substituting into
Eq.~\eqref{eq:fdpc-rep-z-cycle} yields
\[
    (H_AT+V)H_R=0.
\]
The full row rank of $H_R$ implies $V=H_AT$. Hence
\[
    \ker\phi_Z=\operatorname{row}(H_Z).
\]
Since $\phi_Z$ is surjective and
$\ker\phi_Z=\operatorname{row}(H_Z)$, it induces the isomorphism
\[
    \mathcal L_Z\cong\mathcal C_A,
\]
which proves Eq.~\eqref{eq:fdpc-rep-logical-quotient}.

Finally, let $\ba=\phi_Z(U,V)$. Whenever $a_i=1$, the $i$th row of
$U$ has odd weight and hence contains at least one nonzero entry.
Therefore,
\[
    \operatorname{wt}(U,V)
    \geq
    \operatorname{wt}(U)
    \geq
    \operatorname{wt}(\ba).
\]
The representative $(\ba\bee_j^T,0)$ attains equality for any
$j\in[r]$. Thus, the minimum physical weight of the logical class
corresponding to $\ba$ is exactly $\operatorname{wt}(\ba)$, which
proves Eq.~\eqref{eq:fdpc-rep-exact-logical-distribution}.
\end{proof}

Proposition~\ref{prop:fdpc-rep-logical-distribution} gives an exact transfer
of the classical FDPC weight distribution to the $Z$-logical classes of the
qFDPC--Rep code. In particular, the FDPC weight distribution is not just
the weight distribution of a selected set of logical representatives. Instead, it gives the minimum physical weight of every logical equivalence class in this sector. For a random FDPC component,
\begin{equation}
    \mathbb E[M_w^Z]
    =
    \mathbb E[\mathcal A_w].
    \label{eq:average-logical-class-distribution}
\end{equation}
Hence, the ensemble-average weight distribution derived earlier applies to the $Z$-logical-class weight distribution of the corresponding
qFDPC--Rep ensemble.

\subsection{First-order approximation of the ML probability of error}
\label{subsec:first-order-ml}

Since errors that differ only by a stabilizer are equivalent, ML decoding over a fixed erasure pattern can fail logically only when the erased qubits support a nontrivial logical operator \cite{Grassl1997QuantumErasure,
Pech2026StabilizerInactivation}. This standard characterization makes the
minimum-weight logical operators particularly important in the low-erasure
regime. If $D_Z$ denotes the $Z$ distance and $N_{D_Z}^{Z}$ is the number
of distinct physical $Z$-logical operators of weight $D_Z$, then we have
\begin{equation}
    P_{\mathrm e,Z}^{\mathrm{ML}}(p)
    =
    \frac{N_{D_Z}^{Z}}{2}p^{D_Z}
    +O\!\left(p^{D_Z+1}\right).
    \label{eq:first-order-general-z}
\end{equation}
In fact, no erasure set of cardinality smaller than $D_Z$ can support a
nontrivial $Z$-logical operator. For an erasure set of cardinality exactly
$D_Z$, the presence of a minimum-weight logical operator produces one
binary logical ambiguity and hence a conditional ML error probability
$1/2$. A similar statement holds for the $X$ sector.

For HGP codes, the multiplicities $N_{D_Z}^{Z}$ and $N_{D_X}^{X}$ can be
related explicitly to the minimum-weight distributions of the classical
component codes. We assume throughout that the classical component codes have full
coordinate support, i.e., for every coordinate $j$ there exists a
codeword whose $j$th coordinate is nonzero. Equivalently, a generator
matrix of each component code has no all-zero column. Similarly, for a parity-check matrix, full coordinate support means
that no column is linearly independent of all the remaining columns.
The repetition component considered below has full coordinate support,
and we restrict the FDPC-based logical-weight analysis to realizations
with this property. The condition is verifiable for any
particular realization and is used only in counting the multiplicity of
the line logical operators.

\begin{proposition}[First-order ML erasure probability of HGP codes]
\label{prop:hgp-first-order-erasure}
Let $H_A$ and $H_B$ be full-row-rank parity-check matrices of classical
codes $\mathcal C_A$ and $\mathcal C_B$, with lengths $n_A$ and $n_B$,
minimum distances $d_A$ and $d_B$, and full coordinate support. Let
$A_{d_A}^{(A)}$ and $A_{d_B}^{(B)}$ denote the numbers of minimum-weight
codewords in the two component codes. Then, for the corresponding HGP
code, the sector distances satisfy
$D_Z=d_A$ and $D_X=d_B$. Moreover,
\begin{align}
    P_{\mathrm e,Z}^{\mathrm{ML}}(p)
    &=
    \frac{n_B A_{d_A}^{(A)}}{2}
    p^{d_A}
    +O\!\left(p^{d_A+1}\right),
    \label{eq:hgp-first-order-z}\\
    P_{\mathrm e,X}^{\mathrm{ML}}(p)
    &=
    \frac{n_A A_{d_B}^{(B)}}{2}
    p^{d_B}
    +O\!\left(p^{d_B+1}\right).
    \label{eq:hgp-first-order-x}
\end{align}
\end{proposition}

\begin{proof}
We prove the $Z$-sector statement and the $X$-sector statement follows
by symmetry. As discussed in the preceding subsection, the logical
operators of an HGP code admit line representatives determined by
codewords of the classical components. Under the
full-coordinate-support assumption, for every minimum-weight codeword
$\ba\in\mathcal C_A$ and every $j\in[n_B]$,
$(\ba\bee_j^T,0)$ is a nontrivial $Z$-logical operator of weight
$d_A$.

Conversely, applying the same kernel argument used in the proof of
Proposition~\ref{prop:fdpc-rep-logical-distribution} to the second component shows that for every nontrivial
$Z$-logical operator $(U,V)$ there exists
$\bbb\in\mathcal C_B$ such that $0\neq U\bbb\in\mathcal C_A$.
Consequently,
\[
    d_A
    \leq \operatorname{wt}(U\bbb)
    \leq \operatorname{wt}(U)
    \leq \operatorname{wt}(U,V),
\]
and hence $D_Z=d_A$.

It remains to characterize the equality case. If
$\operatorname{wt}(U,V)=d_A$, then the above inequalities imply
$V=0$ and $\operatorname{wt}(U)=d_A$. The logical constraint then
reduces to $H_AU=0$ and, hence, every column of $U$ belongs to
$\mathcal C_A$. Since every nonzero column has weight at least $d_A$
and the total weight of $U$ is exactly $d_A$, $U$ contains exactly
one nonzero column. Thus every minimum-weight $Z$-logical operator is
of the form $(\ba\bee_j^T,0)$ with
$\operatorname{wt}(\ba)=d_A$. Therefore,
\begin{equation}
    N_{d_A}^Z
    =
    n_B A_{d_A}^{(A)}.
\end{equation}

No erasure pattern of weight smaller than $d_A$ can support a
nontrivial $Z$-logical operator. At weight $d_A$, each
minimum-weight logical support produces a single binary logical
ambiguity and, hence, a conditional ML error probability of $1/2$.
Thus,
\[
    P_{\mathrm e,Z}^{\mathrm{ML}}(p)
    =
    \frac{n_B A_{d_A}^{(A)}}{2}p^{d_A}
    +O(p^{d_A+1}),
\]
completing the proof. 
\end{proof}

For the symmetric qFDPC construction of
Definition~\ref{def:qfdpc}, the two component codes are identical. Let
$n$, $d$, and $A_d$ denote the component blocklength, minimum distance,
and number of minimum-weight codewords, respectively. Proposition~
\ref{prop:hgp-first-order-erasure} gives
\begin{equation}
    P_{\mathrm e,Z}^{\mathrm{ML}}(p)
    =
    P_{\mathrm e,X}^{\mathrm{ML}}(p)
    =
    \frac{nA_d}{2}p^d
    +O\!\left(p^{d+1}\right).
    \label{eq:qfdpc-first-order-sectors}
\end{equation}
For the FDPC codes considered here, $d\geq4$, and the minimum-weight
$X$- and $Z$-logical operators have distinct supports. The two sector
contributions therefore add at first order, yielding
\begin{equation}
    P_{\mathrm e}^{\mathrm{ML}}(p)
    =
    nA_d\,p^d
    +O\!\left(p^{d+1}\right).
    \label{eq:qfdpc-first-order-ml}
\end{equation}
Thus, the first-order ML probability of error of a qFDPC code is
determined specifically by the minimum-distance term of the classical FDPC
weight distribution.

For the qFDPC--Rep construction of
Definition~\ref{def:qfdpc-rep}, let the FDPC component have parameters
$(n,d)$ and minimum-weight multiplicity $A_d$. The repetition component
has length $r$, minimum distance $r$, and a single nonzero codeword.
Proposition~\ref{prop:hgp-first-order-erasure} therefore gives
\begin{align}
    P_{\mathrm e,Z}^{\mathrm{ML}}(p)
    &=
    \frac{rA_d}{2}p^d
    +O\!\left(p^{d+1}\right),
    \label{eq:qfdpc-rep-first-order-z}\\
    P_{\mathrm e,X}^{\mathrm{ML}}(p)
    &=
    \frac{n}{2}p^r
    +O\!\left(p^{r+1}\right).
    \label{eq:qfdpc-rep-first-order-x}
\end{align}
Consequently,
\begin{equation}
    P_{\mathrm e}^{\mathrm{ML}}(p)
    =
    \begin{cases}
        \dfrac{rA_d}{2}p^d+O(p^{d+1}),
        & d<r,\\[1.2ex]
        \dfrac{n}{2}p^r+O(p^{r+1}),
        & r<d,\\[1.2ex]
        \dfrac{rA_d+n}{2}p^d+O(p^{d+1}),
        & r=d.
    \end{cases}
    \label{eq:qfdpc-rep-first-order-total}
\end{equation}
Note that when $d=r$, the minimum-weight $Z$- and $X$-line operators have distinct physical supports as the former are column-line operators while the latter are row-line operators. Hence, their first-order contributions add.

Equations~\eqref{eq:qfdpc-first-order-ml} and
\eqref{eq:qfdpc-rep-first-order-total} provide an explicit operational use
of the FDPC weight distribution. In particular, the minimum distance
determines the leading exponent of the ML probability of error, while the
minimum-weight multiplicity $A_d$ determines its first-order coefficient.

\noindent
{\bf Remark 2.} The connection between the FDPC weight distribution and ML
performance developed above relies on the erasure locations being known
to the decoder. For standard Pauli code-capacity channels, degenerate ML
decoding compares probabilities of stabilizer cosets, and the scalar
weight distributions of the classical component codes do not in general
determine these coset probabilities. Therefore, extending our analysis to
such channels requires additional information beyond the
component weight distributions.


\section{Finite-Length qLDPC Constructions from FDPC Components}
\label{sec:finite-qfdpc}

In this section, we use the finite-length FDPC results developed throughout the paper
to construct qLDPC codes with bounded stabilizer weight. Our objective is not
to optimize a single code parameter in isolation. Instead, we aim to illustrate
the tradeoffs among quantum blocklength, dimension, minimum distance, and
stabilizer weight that can be obtained using FDPC components in HGP
constructions. We focus on quantum blocklengths below $10^5$ and consider
two stabilizer-weight regimes. We first restrict the maximum stabilizer
weight to $10$, enabling comparison with commonly studied finite-length
qLDPC codes with small checks. We then moderately relax this constraint to
illustrate the substantial rate improvements that become available with
larger stabilizer weights. The component-code distance guarantees and
weight-distribution quantities used throughout this section are obtained
from the finite-length FDPC analysis developed earlier in the paper.

\subsection{Constructions with stabilizer weight up to 10}
\label{subsec:finite-weight10}

We first consider two classes of finite-length qLDPC constructions with
maximum stabilizer weight at most $10$. We refer to the symmetric HGP construction using the same
order-$2$ sparsified FDPC code as both classical
components as a \emph{sim-FD qLDPC} code. We refer to the HGP construction
using an order-$2$ sparsified FDPC code as one component and a repetition code
as the other as an \emph{FD--Rep qLDPC} code. These names will be used
throughout this section to distinguish the finite-length qLDPC
constructions from the general qFDPC families introduced earlier.

Specifically, we use the order-$2$ sparsified FDPC components of
Subsection~\ref{subsec:order2-sparsified}, for which the ensemble-average
weight distribution is available explicitly from
Eq.~\eqref{eq:order2-sparsified-wd}. Such a component has row
weight $L$ and column weight $4$. Consequently, the stabilizer weight
$w_{\rm stab}$ for the sim-FD qLDPC construction satisfies
\begin{equation}
    w_{\rm stab}\leq L+4,
\end{equation}
whereas the FD--Rep qLDPC construction satisfies
\[
    w_{\rm stab}\leq \max\{L+2,2s+2\}.
\]
For $L\geq4$ and $s=2$, as in all constructions reported below, this
reduces to
\begin{equation}
    w_{\rm stab}\leq L+2.
\end{equation}
Thus, the constraint $w_{\rm stab}\leq10$ permits $L\leq6$ for sim-FD
qLDPC codes and $L\leq8$ for FD--Rep qLDPC codes. We restrict attention
to quantum blocklengths $N<10^5$.

To enhance the analysis beyond the minimum distance values, we use the complete analytical FDPC
weight distribution to obtain a higher-order approximation to the ML logical
block error probability. We refer to this as the \emph{weight-distribution
(WD) approximation}. Recall from
Section~\ref{sec:erasure-distribution} that the first-order ML probability
is determined by the minimum-weight logical operators. The same
line-logical construction produces logical operators at all weights
appearing in the component-code weight distribution. The WD approximation
retains all of these contributions rather than only the minimum-weight
term.

Following the first-moment argument used in the proof of
Theorem~\ref{thm:fdpc-distance}, define
\begin{equation}
    \delta_d
    \triangleq
    \sum_{w=1}^{d-1}\overline A_w,
  \ \text{where}\ 
    \overline A_w
\triangleq
\mathbb E[\mathcal A_w]
=
\frac{B_{q,L,w}^2}{\binom{q^2}{w}}.
    \label{eq:finite-delta-d}
\end{equation}
Thus, $\overline A_w$ is the ensemble-average weight distribution given
explicitly by Eq.~\eqref{eq:order2-sparsified-wd}. Throughout this section, every FDPC component is drawn from the
order-$2$ sparsified ensemble of
Subsection~\ref{subsec:order2-sparsified}. As in Proposition~\ref{prop:hgp-first-order-erasure}, the line-logical multiplicities used below
are understood for full-coordinate-support realizations. 

Let $d_A$ denote the minimum distance of a realization $A$ of the
order-$2$ sparsified FDPC ensemble of
Subsection~\ref{subsec:order2-sparsified}. Since the event $d_A<d$, for a given $d$,
requires the existence of at least one nonzero codeword of weight below
$d$, the first-moment bound gives
\begin{equation}
    \Pr(d_A<d)\leq\delta_d,
\end{equation}
and hence
\begin{equation}
    \Pr(d_A\geq d)\geq1-\delta_d.
\end{equation}

For the finite-length constructions reported in this section, we define
\begin{equation}
    d_A^{\rm cert}
    \triangleq
    \max\left\{
        d:\delta_d\leq\frac{1}{2}
    \right\},
    \label{eq:finite-component-distance}
\end{equation}
which implies
\begin{equation}
    \Pr\!\left(
        d_A\geq d_A^{\rm cert}
    \right)
    \geq \frac{1}{2}.
    \label{eq:finite-component-distance-prob}
\end{equation}
Thus, for every single FDPC component parameter choice considered
later in this section, there exist realizations attaining the reported
certified component distance. Also, the subensemble of these realizations is called the \textit{good-distance}
subensemble. Therefore, in forming the WD
approximation for this subensemble, we set the
weight-distribution terms below $d_A^{\rm cert}$ to zero and retain
the analytically available ensemble-average coefficients
$\overline A_w$ for $w\geq d_A^{\rm cert}$. This truncation is an approximation to the weight distribution of the
good-distance subensemble, since $\overline A_w$ denotes the
unconditional ensemble average rather than the conditional quantity
\[
    \mathbb E\!\left[
        \cA_w
        \,\middle|\,
        d_A\geq d_A^{\rm cert}
    \right].
\]
Nevertheless, Eq.~\eqref{eq:finite-component-distance-prob} implies the one-sided bound
\[
    \mathbb E\!\left[
        \cA_w
        \,\middle|\,
        d_A\geq d_A^{\rm cert}
    \right]
    \leq
    2\overline A_w,
\]
for every $w$. Thus, conditioning on the certified-distance event can
increase any retained ensemble-average coefficient by at most a factor
of two. 

For a sim-FD qLDPC code, the resulting weight-distribution (WD)
approximation to the ML logical block error probability at physical
erasure probability $p$ is
\begin{equation}
    \widehat P_{\rm e}^{\rm WD}(p)
    =
    n\sum_{w=d_A^{\rm cert}}^{n}
    \overline A_w p^w ,
    \label{eq:qfdpc-wd-approx}
\end{equation}
where $n$ is the blocklength of the FDPC component. For an FD--Rep qLDPC
code with repetition length $r$, we use
\begin{equation}
    \widehat P_{\rm e}^{\rm WD}(p)
    =
    \frac{r}{2}
    \sum_{w=d_A^{\rm cert}}^{n}
    \overline A_w p^w
    +
    \frac{n}{2}p^r .
    \label{eq:qfdpc-rep-wd-approx}
\end{equation}

The leading terms of Eqs.~\eqref{eq:qfdpc-wd-approx} and
\eqref{eq:qfdpc-rep-wd-approx} coincide with the rigorous first-order ML
expressions of Section~\ref{sec:erasure-distribution}, and the remaining
terms retain the higher-weight line logical operators determined by the
complete FDPC weight distribution. Two effects are not captured: logical
operators outside the line-logical family, and overlaps among distinct
logical-erasure events.

For later use, let $D^{\rm cert}$ denote the certified quantum distance
of the construction. For the sim-FD qLDPC and FD--Rep qLDPC codes considered here,
\begin{equation}
    D^{\rm cert}
    =
    \begin{cases}
        d_A^{\rm cert}, & \text{sim-FD qLDPC},\\
        \min\{d_A^{\rm cert},r\}, & \text{FD--Rep qLDPC}.
    \end{cases}
    \label{eq:finite-certified-distance}
\end{equation}
Accordingly, the leading coefficient of the WD approximation is
\begin{equation}
    C_{D^{\rm cert}}
    =
    \begin{cases}
        n\overline A_{d_A^{\rm cert}},
        & \text{sim-FD qLDPC},\\[1ex]
        \dfrac{r}{2}\overline A_{d_A^{\rm cert}},
        & \text{FD--Rep qLDPC},\quad d_A^{\rm cert}<r,\\[2ex]
        \dfrac{n}{2},
        & \text{FD--Rep qLDPC},\quad r<d_A^{\rm cert},\\[2ex]
        \dfrac{r\overline A_r+n}{2},
        & \text{FD--Rep qLDPC},\quad r=d_A^{\rm cert}.
    \end{cases}
    \label{eq:wd-leading-coefficient}
\end{equation}

\begin{table*}[t]
\centering
\caption{Representative finite-length sim-FD qLDPC and FD--Rep qLDPC constructions
with $N<10^5$ and maximum stabilizer weight $w_{\rm stab}\leq10$.
Here $d_A^{\rm cert}$ is the certified FDPC component distance defined in Eq.~\eqref{eq:finite-component-distance}, and $D^{\rm cert}$ is the
corresponding certified quantum distance defined in Eq.~\eqref{eq:finite-certified-distance}. For FD--Rep qLDPC codes,
$r$ denotes the repetition-code length. WD denotes the
weight-distribution approximation defined in
Eqs.~\eqref{eq:qfdpc-wd-approx}--\eqref{eq:qfdpc-rep-wd-approx}.
The error-floor column gives the estimated logical error probability at
which the leading retained contribution at $D^{\rm cert}$ becomes
equal to all higher-weight contributions combined. The quantity $p_{\rm cap}$ is the
quantum-erasure capacity boundary evaluated at the guaranteed rate. Values marked by $\dagger$ lie beyond the reported asymptotic
capacity reference and should be interpreted only as formal
crossings of the WD approximation.}
\label{tab:qfdpc-weight10}

\resizebox{\textwidth}{!}{%
\renewcommand{\arraystretch}{1.35}
\begin{tabular}{c|c|c|r|r|c|c|c|c|c|c|c|c}
\hline
\rule{0pt}{3.2ex}
Construction
& $(q,L)$
& $r$
& $N$
& $K$
& $R_Q$
& $d_A^{\rm cert}$
& $D^{\rm cert}$
& $w_{\rm stab}$
& \shortstack{$\widehat P_{\rm e}$ at\\error-floor onset}
& $p_{\rm ML}^{\rm WD}(10^{-6})$
& $p_{\rm ML}^{\rm WD}(10^{-12})$
& $p_{\rm cap}$
\\[0.7ex]
\hline

sim-FD qLDPC
& $(10,5)$ & -- & $\leq15\,184$ & $\geq784$
& $\geq5.16\%$ & $14$ & $14$ & $9$
& $3.48\times10^{-5}$
& $0.284$ & $0.108$ & $0.474$
\\

sim-FD qLDPC
& $(12,6)$ & -- & $\leq28\,480$ & $\geq3\,136$
& $\geq11.01\%$ & $14$ & $14$ & $10$
& $2.89\times10^{-5}$
& $0.256$ & $0.098$ & $0.445$
\\

sim-FD qLDPC
& $(15,5)$ & -- & $\leq76\,869$ & $\geq3\,969$
& $\geq5.16\%$ & $32$ & $32$ & $9$
& $1.15\times10^{-11}$
& $0.498^\dagger$ & $0.351$ & $0.474$
\\

sim-FD qLDPC
& $(12,4)$ & -- & $\leq36\,612$ & $\geq324$
& $\geq0.885\%$ & $38$ & $38$ & $8$
& $1.30\times10^{-12}$
& $0.591^\dagger$ & $0.427$ & $0.496$
\\
\hline

FD--Rep qLDPC
& $(64,8)$ & $16$ & $\leq94\,336$ & $\geq2\,176$
& $\geq2.31\%$ & $146$ & $16$ & $10$
& $7.43\times10^{-4}$
& $0.262$ & $0.110$ & $0.488$
\\

FD--Rep qLDPC
& $(56,8)$ & $22$ & $\leq99\,862$ & $\geq1\,666$
& $\geq1.67\%$ & $114$ & $22$ & $10$
& $1.33\times10^{-6}$
& $0.375$ & $0.204$ & $0.492$
\\

FD--Rep qLDPC
& $(48,8)$ & $29$ & $\leq97\,056$ & $\geq1\,224$
& $\geq1.26\%$ & $84$ & $29$ & $10$
& $5.53\times10^{-10}$
& $0.372$ & $0.302$ & $0.494$
\\

FD--Rep qLDPC
& $(40,8)$ & $42$ & $\leq97\,950$ & $\geq850$
& $\geq0.868\%$ & $60$ & $42$ & $10$
& $6.94\times10^{-17}$
& $0.367$ & $0.355$ & $0.496$
\\

FD--Rep qLDPC
& $(35,7)$ & $53$ & $\leq98\,725$ & $\geq575$
& $\geq0.582\%$ & $64$ & $53$ & $9$
& $1.46\times10^{-19}$
& $0.422$ & $0.406$ & $0.497$
\\

FD--Rep qLDPC
& $(30,6)$ & $69$ & $\leq99\,500$ & $\geq350$
& $\geq0.352\%$ & $70$ & $69$ & $8$
& $4.16\times10^{-23}$
& $0.495$ & $0.472$ & $0.498$
\\
\hline
\end{tabular}%
}
\end{table*}

For a target logical block error probability $\varepsilon$, we define
$p_{\rm ML}^{\rm WD}(\varepsilon)$ implicitly by
\begin{equation}
    \widehat P_{\rm e}^{\rm WD}
    \left(
        p_{\rm ML}^{\rm WD}(\varepsilon)
    \right)
    =
    \varepsilon .
    \label{eq:pml-wd-target}
\end{equation}
Thus, $p_{\rm ML}^{\rm WD}(\varepsilon)$ is the physical erasure
probability at which the WD approximation predicts an ML logical block
error probability $\varepsilon$. For the constructions with
$w_{\rm stab}\leq10$, we report
$p_{\rm ML}^{\rm WD}(10^{-6})$ and
$p_{\rm ML}^{\rm WD}(10^{-12})$. The latter probes an ultra-low-error
regime that is difficult to access by Monte Carlo simulation,
while the two operating points together indicate the steepness of the
predicted low-error behavior.

Waterfall and error-floor behavior has also been observed
numerically for long HGP codes over the quantum erasure
channel~\cite{GokdumanYaoPfister2025HGP}. Here, we use the
analytically available WD approximation to characterize the
onset of the error-floor regime. Writing
\begin{equation}
    \widehat P_{\rm e}^{\rm WD}(p)
    =
    C_{D^{\rm cert}}p^{D^{\rm cert}}
    +
    \widehat P_{\rm e,>D^{\rm cert}}^{\rm WD}(p),
    \label{eq:wd-min-higher-split}
\end{equation}
we define the error-floor onset $p_{\rm EF}$ implicitly by
\begin{equation}
    C_{D^{\rm cert}}p_{\rm EF}^{D^{\rm cert}}
    =
    \widehat P_{\rm e,>D^{\rm cert}}^{\rm WD}(p_{\rm EF}).
    \label{eq:error-floor-onset}
\end{equation}
At $p=p_{\rm EF}$, the leading retained contribution at
$D^{\rm cert}$ accounts for one half of the WD approximation;
for $p<p_{\rm EF}$, this is the dominant retained term. Rather than reporting $p_{\rm EF}$ itself,
Table~\ref{tab:qfdpc-weight10} reports
$\widehat P_{\rm e}^{\rm WD}(p_{\rm EF})$, which gives the estimated
logical block error probability at which the leading-term-dominated error-floor regime begins.

For reference, the quantum capacity of the qubit erasure channel is
\cite{Bennett1997QuantumErasure}
\begin{equation}
    Q_{\rm er}(p)
    =
    \max\{1-2p,0\}.
\end{equation}
The asymptotic quantum-capacity boundary corresponding to rate $R_Q$ is
\begin{equation}
    p_{\rm cap}
    \triangleq
    \frac{1-R_Q}{2}.
    \label{eq:erasure-capacity-probability}
\end{equation}
The values of $p_{\rm cap}$ reported below are evaluated using the
guaranteed finite-length rate. 
Since the actual dimension of a
realization may exceed its guaranteed lower bound, the reported
$p_{\rm cap}$ is an upper bound on the capacity boundary corresponding
to the actual code rate. Moreover, for qLDPC families with uniformly bounded
stabilizer-generator weight, stricter asymptotic rate upper
bounds over the quantum erasure channel are known
\cite{DelfosseZemor2013ErasureBound}. Thus, $p_{\rm cap}$ is
included only as a general information-theoretic reference.

Table~\ref{tab:qfdpc-weight10} illustrates two separate
finite-length regimes. The sim-FD qLDPC construction provides the
larger rates; for example, the $(q,L)=(12,6)$ construction has guaranteed
rate above $11\%$ with stabilizer weight $10$. The FD--Rep qLDPC construction
instead trades rate for certified quantum distance: under the same
$N<10^5$ and $w_{\rm stab}\leq10$ constraints, the examples shown reach
$D^{\rm cert}=69$ while maintaining stabilizer weight between $8$ and
$10$.

The error-floor-onset column shows a strong dependence on
$D^{\rm cert}$. For the smaller-distance constructions, the
leading retained contribution becomes dominant at comparatively moderate
logical block error probabilities. For example, the FD--Rep qLDPC code with
$D^{\rm cert}=16$ enters the minimum-weight-dominated regime at
$\widehat P_{\rm e}^{\rm WD}\simeq7.4\times10^{-4}$, while for
$D^{\rm cert}=22$ the corresponding value is approximately
$1.3\times10^{-6}$. As the certified distance increases, the onset is
rapidly pushed to much lower error probabilities: it is approximately
$5.5\times10^{-10}$ for $D^{\rm cert}=29$ and below $10^{-16}$ for
$D^{\rm cert}=42$.

The two $p_{\rm ML}^{\rm WD}$ columns provide an alternative view of the
predicted low-error behavior. For the smaller-distance codes, reducing the
target logical block error probability from $10^{-6}$ to $10^{-12}$
requires a substantial reduction in the physical erasure probability. For
example, for the $(q,L)=(12,6)$ sim-FD qLDPC construction,
$p_{\rm ML}^{\rm WD}$ decreases from approximately $0.256$ to $0.098$.
For the larger-distance FD--Rep qLDPC constructions, the two
operating points become much closer, which reflects the strong dependence
of the complete retained WD sum on the physical erasure probability.
At these operating points, higher-weight line-logical terms can also
make a substantial contribution.

The higher-weight terms can also be important when the certified quantum
distance of an FD--Rep qLDPC code is determined by the repetition component,
i.e., when $r<d_A^{\rm cert}$. In this case, the repetition term determines
the leading power of $p$, while the accumulation of higher-weight FDPC
line logicals can become significant at larger erasure probabilities.
Consequently, the complete FDPC weight distribution provides finite-length
information beyond the certified distance and its leading multiplicity.

\subsection{Constructions with moderately relaxed stabilizer weight}
\label{subsec:finite-relaxed}

We next moderately relax the stabilizer-weight constraint and consider
constructions with
\[
    11\leq w_{\rm stab}\leq16.
\]
Increasing the sparsification parameter $L$ reduces the rate loss of the
classical FDPC components and can therefore substantially increase the
rate of the resulting quantum code. In this regime, we consider the
sim-FD and FD--Rep qLDPC constructions introduced in the preceding
subsection. We also consider asymmetric HGP constructions using two
different order-$2$ sparsified FDPC components, which we refer
to as FD--FD qLDPC codes. Allowing the two FDPC components to have
different parameters provides an additional knob to tune for the finite-length rate--distance
tradeoff.

For an FD--FD qLDPC construction with component codes $A$ and $B$,
let $n_A$ and $n_B$ denote their blocklengths, and let
$d_A^{\rm cert}$ and $d_B^{\rm cert}$ denote their certified component
distances as defined in Eq.~\eqref{eq:finite-component-distance}, using
the first-moment criterion $\delta_d\leq1/2$. The corresponding certified
quantum distance is
\begin{equation}
    D^{\rm cert}
    =
    \min\left\{
        d_A^{\rm cert},
        d_B^{\rm cert}
    \right\}.
    \label{eq:asymmetric-certified-quantum-distance}
\end{equation}
Note that the two FDPC components are drawn independently. Since each
component satisfies its certified-distance bound with probability
at least $1/2$, both bounds hold simultaneously with probability
at least $1/4$. Consequently, there exists an FD--FD realization
satisfying
\[
D \geq
\min\left\{
d_A^{\rm cert},d_B^{\rm cert}
\right\}
=
D^{\rm cert}.
\]

Let $\overline A_w^{(A)}$ and $\overline A_w^{(B)}$ denote the
ensemble-average weight distributions of the two FDPC components. Using
the same good-distance truncation convention as in the preceding
subsection, the weight-distribution (WD) approximation extends to
\begin{equation}
    \widehat P_{\rm e}^{\rm WD}(p)
    =
    \frac{n_B}{2}
    \sum_{w=d_A^{\rm cert}}^{n_A}
        \overline A_w^{(A)}p^w
    +
    \frac{n_A}{2}
    \sum_{w=d_B^{\rm cert}}^{n_B}
        \overline A_w^{(B)}p^w .
    \label{eq:asymmetric-fdpc-wd-approx}
\end{equation}
The two terms correspond to the two families of line logical operators.
For identical components,
Eq.~\eqref{eq:asymmetric-fdpc-wd-approx} reduces to
Eq.~\eqref{eq:qfdpc-wd-approx}. We use the same definitions of
$p_{\rm ML}^{\rm WD}(\varepsilon)$ and the error-floor onset as in the
preceding subsection.

\begin{table*}[t]
\centering
\caption{Representative finite-length sim-FD, FD--FD, and FD--Rep qLDPC
constructions with $N<10^5$ and moderately relaxed stabilizer weight
$11\leq w_{\rm stab}\leq16$. A component $(q,L)$ denotes an order-$2$ sparsified FDPC code of
Subsection~\ref{subsec:order2-sparsified}, and $\operatorname{Rep}(r)$ denotes a
length-$r$ repetition code. The FDPC component distances are certified
with $\delta_d\leq1/2$, and $D^{\rm cert}$ denotes the resulting certified
quantum distance defined by Eq.~\eqref{eq:finite-certified-distance} for the sim-FD and FD--Rep constructions and
by Eq.\,\eqref{eq:asymmetric-certified-quantum-distance} for the FD--FD constructions. WD denotes the weight-distribution approximation.
The error-floor column gives the estimated logical block error probability
at which the leading retained contribution at $D^{\rm cert}$ becomes
equal to all higher-weight contributions combined. The quantities
$p_{\rm ML}^{\rm WD}(10^{-3})$ and
$p_{\rm ML}^{\rm WD}(10^{-6})$ are the physical erasure probabilities at
which the WD approximation reaches the corresponding target logical block
error probabilities. The quantity $p_{\rm cap}$ is the quantum-erasure
capacity boundary evaluated at the guaranteed rate.}
\label{tab:qfdpc-relaxed}

\resizebox{\textwidth}{!}{%
\renewcommand{\arraystretch}{1.35}
\begin{tabular}{c|c|r|r|c|c|c|c|c|c|c}
\hline
\rule{0pt}{3.2ex}
Construction
& Components
& $N$
& $K$
& $R_Q$
& $D^{\rm cert}$
& $w_{\rm stab}$
& \shortstack{$\widehat P_{\rm e}$ at\\error-floor onset}
& $p_{\rm ML}^{\rm WD}(10^{-3})$
& $p_{\rm ML}^{\rm WD}(10^{-6})$
& $p_{\rm cap}$
\\[0.7ex]
\hline

sim-FD qLDPC
& $(14,7)\times(14,7)$
& $\leq49\,232$ & $\geq8\,464$
& $\geq17.19\%$ & $12$ & $11$
& $1.58\times10^{-4}$
& $0.339$ & $0.211$ & $0.414$
\\

sim-FD qLDPC
& $(16,8)\times(16,8)$
& $\leq79\,936$ & $\geq18\,496$
& $\geq23.14\%$ & $12$ & $12$
& $1.13\times10^{-4}$
& $0.307$ & $0.198$ & $0.384$
\\

FD--FD qLDPC
& $(16,8)\times(18,9)$
& $\leq99\,264$ & $\geq25\,568$
& $\geq25.76\%$ & $12$ & $13$
& $6.60\times10^{-5}$
& $0.283$ & $0.193$ & $0.371$
\\

FD--FD qLDPC
& $(10,5)\times(27,9)$
& $\leq94\,932$ & $\geq11\,844$
& $\geq12.48\%$ & $14$ & $13$
& $1.30\times10^{-5}$
& $0.315$ & $0.260$ & $0.438$
\\

FD--FD qLDPC
& $(8,4)\times(33,11)$
& $\leq90\,864$ & $\geq5\,688$
& $\geq6.26\%$ & $18$ & $15$
& $4.53\times10^{-9}$
& $0.258$ & $0.249$ & $0.469$
\\

FD--FD qLDPC
& $(12,4)\times(21,7)$
& $\leq92\,988$ & $\geq3\,726$
& $\geq4.01\%$ & $24$ & $11$
& $3.41\times10^{-10}$
& $0.404$ & $0.381$ & $0.480$
\\
\hline

FD--Rep qLDPC
& $(36,9)\times\operatorname{Rep}(38)$
& $\leq69\,376$ & $\geq752$
& $\geq1.08\%$ & $38$ & $11$
& $4.69\times10^{-18}$
& $0.325$ & $0.318$ & $0.495$
\\

FD--Rep qLDPC
& $(40,10)\times\operatorname{Rep}(36)$
& $\leq78\,880$ & $\geq992$
& $\geq1.26\%$ & $36$ & $12$
& $1.98\times10^{-18}$
& $0.291$ & $0.286$ & $0.494$
\\

FD--Rep qLDPC
& $(48,12)\times\operatorname{Rep}(33)$
& $\leq99\,584$ & $\geq1\,568$
& $\geq1.57\%$ & $33$ & $14$
& $5.24\times10^{-19}$
& $0.241$ & $0.237$ & $0.492$
\\

FD--Rep qLDPC
& $(52,13)\times\operatorname{Rep}(28)$
& $\leq97\,312$ & $\geq1\,904$
& $\geq1.96\%$ & $28$ & $15$
& $2.11\times10^{-16}$
& $0.222$ & $0.219$ & $0.490$
\\

FD--Rep qLDPC
& $(56,14)\times\operatorname{Rep}(25)$
& $\leq99\,136$ & $\geq2\,272$
& $\geq2.29\%$ & $25$ & $16$
& $5.44\times10^{-15}$
& $0.205$ & $0.203$ & $0.489$
\\
\hline
\end{tabular}%
}
\end{table*}

The constructions considered here have substantially higher rates than
many of the distance-oriented codes in
Table~\ref{tab:qfdpc-weight10}. We therefore report the operating points
$p_{\rm ML}^{\rm WD}(10^{-3})$ and
$p_{\rm ML}^{\rm WD}(10^{-6})$, rather than the $10^{-6}$ and
$10^{-12}$ values used there. These target logical block error
probabilities better characterize the operating regime of the higher-rate
codes. The error-floor-onset column is retained separately. The numbers reported in this column indicate the
lower-error regime in which the minimum-weight contribution becomes
dominant.

Table~\ref{tab:qfdpc-relaxed} shows that even a moderate relaxation of
the stabilizer-weight constraint can produce a substantial increase in
rate. For the sim-FD qLDPC construction, allowing stabilizer
weight $12$ yields a guaranteed rate above $23\%$. Allowing the two FDPC
components to differ provides a further improvement: the
$(16,8)\times(18,9)$ FD--FD qLDPC construction reaches a guaranteed rate
of approximately $25.8\%$ with stabilizer weight $13$ and certified
quantum distance $D^{\rm cert}=12$. This is more than twice the largest
guaranteed rate among the weight-$10$ sim-FD qLDPC examples in
Table~\ref{tab:qfdpc-weight10}.

The FD--FD qLDPC constructions illustrate the additional rate--distance
flexibility obtained by choosing the two classical components
independently. Within the same $N<10^5$ range, the examples shown span
guaranteed rates from approximately $25.8\%$ at
$D^{\rm cert}=12$ to approximately $4.0\%$ at
$D^{\rm cert}=24$. Within the WD approximation, the estimated error-floor onset is
pushed to lower logical block error probabilities, from roughly
$10^{-4}$ for the distance-$12$ examples to the
$10^{-9}$--$10^{-10}$ range for the distance-$18$ and
distance-$24$ examples.

Another type of tradeoff can be observed within the FD--Rep qLDPC family.
Increasing $L$ from $9$ to $14$ raises the guaranteed quantum rate from
approximately $1.1\%$ to $2.3\%$, while increasing the stabilizer weight
from $11$ to $16$. Over the same sequence,
$D^{\rm cert}$ decreases from $38$ to $25$. The $p_{\rm ML}^{\rm WD}(10^{-3})$ and
$p_{\rm ML}^{\rm WD}(10^{-6})$ values remain relatively close for
these larger-distance codes, which is consistent with the strong dependence
of the retained WD approximation on the physical erasure probability. At the same time, the
minimum-weight-dominated error-floor regime is predicted to begin only at
very small logical block error probabilities.

Tables~\ref{tab:qfdpc-weight10} and
\ref{tab:qfdpc-relaxed} together show that the sparsification parameter $L$ and
the choice of HGP components provide specific finite-length design
controls. The strict $w_{\rm stab}\leq10$ regime includes
distance-oriented constructions with very small checks. On the other hand, a modest
increase in stabilizer weight can substantially increase the
rate. In particular, the relaxed-weight FD--FD qLDPC constructions reach
rates above $20\%$ without changing the underlying FDPC construction or
abandoning bounded-weight stabilizers.

\section{Conclusion}
\label{sec:conclusion}

In this paper, we developed a framework for constructing and analyzing
qLDPC codes and newly introduced qFDPC codes from FDPC components. We derived
finite-length weight distributions and probabilistic distance guarantees
for the relevant FDPC ensembles and introduced sparsification mechanisms
that produce bounded-weight HGP constructions. The resulting sim-FD,
FD--FD, and FD--Rep qLDPC codes exhibit broad tradeoffs among rate,
certified distance, blocklength, and stabilizer weight. A main feature of
the construction is that the FDPC weight distribution remains
analytically accessible. This allows us to relate classical component
codewords to logical operators and to develop a weight-distribution
approximation to the ML logical block error probability over the quantum
erasure channel.

The proposed constructions provide an alternative approach to other
structured finite-length qLDPC families, including generalized bicycle
(GB) codes \cite{Panteleev_2021,wang2022distance}, bivariate bicycle
(BB) codes
\cite{Bravyi_2024,wang2024coprime,
postema2025existencecharacterisationbivariatebicycle}, and the recently
introduced univariate bicycle (UB) codes
\cite{Rabeti2026UnivariateBicycle}. These bicycle constructions
exploit algebraic structure in circulant parity-check matrices and illustrate how additional
structure can be leveraged to analyze finite-length properties of qLDPC codes. These structures could potentially be combined with the proposed
approach to improve the finite-length parameters.

Several directions remain open. Deterministic choices of the
sparsification and underlying permutations in the FDPC construction, together with a sharper
characterization of the good-distance subensemble, could further improve
the finite-length constructions and their ML estimates. It will also be
important to evaluate these codes under practical Pauli-noise decoders.
High-performance belief propagation (BP)-based approaches such as BP-OSD
\cite{Panteleev_2021}, BP with guided decimation (BPGD)
\cite{yao2024beliefpropagationdecodingquantum}, and multiple-bases
BP list decoding (MBBP-LD)
\cite{Rabeti2026MBBP} provide natural starting points. Such studies may
also reveal ways to exploit the FDPC graph structure 
in the decoder.

\bibliographystyle{IEEEtran}
{\footnotesize \bibliography{ref}}
\end{document}